\documentclass[preprint,11pt]{elsarticle}
\usepackage{amsmath}
\usepackage{amssymb}
\usepackage{latexsym}
\usepackage{graphicx}
\usepackage{bm}
\usepackage{nicefrac}
\usepackage[usenames]{color}

\input amssym.def
\newsymbol\rtimes 226F
\newfont{\nset}{msbm10}

\newcommand{\figwidth}{0.9\textwidth}

\newtheorem{theo}{Theorem}[section]
\newtheorem{theorem}[theo]{Theorem}

\newtheorem{lemma}[theo]{Lemma}

\journal{Theoretical Computer Science}

\begin{document}

\begin{frontmatter}

\title{Domination number and minimum dominating sets in pseudofractal scale-free web and Sierpi\'nski graph}

\author[lable1,label2]{Liren Shan}

\author[lable1,label2]{Huan Li}

\author[lable1,label2]{Zhongzhi Zhang}
\ead{zhangzz@fudan.edu.cn}

\address[lable1]{School of Computer Science, Fudan
University, Shanghai 200433, China}
\address[label2]{Shanghai Key Laboratory of Intelligent Information
Processing, Fudan University, Shanghai 200433, China}

\begin{abstract}
The minimum dominating set (MDS) problem is a fundamental subject of theoretical computer science, and has found vast applications in different areas, including sensor networks, protein interaction networks, and structural controllability. However, the determination of the size of a MDS and the number of all MDSs in a general network is NP-hard, and it thus makes sense to seek particular graphs for which the MDS problem can be solved analytically. In this paper, we study the MDS problem in the pseudofractal scale-free web and the Sierpi\'nski graph, which have the same number of vertices and edges. For both networks, we determine explicitly the domination number, as well as the number of distinct MDSs. We show that the pseudofractal scale-free web has a unique MDS, and its domination number is only half of that for the Sierpi\'nski graph, which has many MDSs. We argue that the scale-free topology is responsible for the difference of the size and number of MDSs between the two studied graphs, which in turn indicates that power-law degree distribution plays an important role in the MDS problem and its applications in scale-free networks.
\end{abstract}


\end{frontmatter}

\section{Introduction}

A dominating set of a graph $\mathcal{G}$ with vertex set $\mathcal{V}$ is a subset $\mathcal{D}$ of $\mathcal{V}$, such that each vertex in $\mathcal{V}\setminus \mathcal{D}$ is adjacent to at least one vertex belonging to $\mathcal{D}$. We call $\mathcal{D}$ a minimum dominating set (MDS) if it has the smallest cardinality, and the cardinality of a MDS is called the domination number of graph $\mathcal{G}$. The MDS problem has numerous practical applications in different fields~\cite{HaHeSl98}, particularly in networked systems, such as routing on ad hoc wireless networks~\cite{WuLi01,Wu02}, multi-document summarization in sentence graphs~\cite{ShLi10}, and controllability in protein interaction networks~\cite{Wu14}. Recently, from the angle of MDS, structural controllability of complex networks has been addressed~\cite{NaAk12,NaAk13}, providing an alternative new viewpoint of the control problem for large-scale networks~\cite{LiSlBa11,NeVi12}.

Despite the vast applications, solving the MDS problem of a graph is a challenge, because finding a MDS of a graph is NP-hard~\cite{HaHeSl98}. Over the past years, the MDS problem has attracted considerable attention from theoretical computer science~\cite{FoGrPySt08, HeIs12,dadedeMa14, GaHaK15, CoLeLi15}, discrete and combinatorial mathematics~\cite{MaTa96, KaLiMaNo14, HoKaNa10}, as well as statistical physics~\cite{ZhHaZh15}, and continues to be an active object of research~\cite{LiZhShXu16, GoHeKr16}. Extensive empirical research~\cite{Ne03} uncovered that most real networks exhibit the prominent scale-free behavior~\cite{BaAl99}, with the  degree of their vertices following a power-law distribution $P(k) \sim k^{-\gamma}$. Although an increasing number of studies have been focused on MDS problem, related works about MDS problem in scale-free networks are very rare~\cite{NaAk12, MoDeCzSzSzKo14}. In particular, exact results for domination number and the number of minimum dominating sets in a scale-free network are still lacking.

Due to the ubiquity of scale-free phenomenon in realistic networks, unveiling the behavior of minimum dominating sets with respect to power-law degree distribution is important for better understanding the applications of MDS problem in real-life scale-free networks. On the other hand, determining the domination number and enumerating all minimum dominating sets in a generic network are formidable~\cite{HaHeSl98}, it is thus of great interest to find specific scale-free networks for which the MDS problem can be exactly solved~\cite{LoPl86}.

In this paper, we focus on the domination number and the number of minimum dominating sets in a scale-free network, called pseudofractal scale-free web~\cite{DoGoMe02,ZhQiZhXiGu09} and the Sierpi\'nski graph with the same number of vertices and edges. For both networks, we determine the exact domination number and the number of all different minimum dominating sets. The domination number of the pseudofractal scale-free web is only half of that of the Sierpi\'nski graph. In addition, in the pseudofractal scale-free web, there is a unique minimum dominating set, while in Sierpi\'nski graph the number of all different minimum dominating sets grows exponentially with the number of vertices in the graph. We show the root of the difference of minimum dominating sets between studied networks rests with their architecture dissimilarity, with the pseudofractal scale-free web being heterogeneous while the Sierpi\'nski graph being homogeneous.

\section{Domination number and minimum dominating sets in pseudofractal scale-free web}

In this section, we determine the domination number in the pseudofractal scale-free web, and show that its minimum dominating set is unique.

\subsection{Network construction and properties}

The pseudofractal scale-free web~\cite{DoGoMe02,ZhQiZhXiGu09} considered here is constructed using an iterative approach. Let  $\mathcal{G}_n$, $n\geq 1$, denote the $n$-generation network. Then the scale-free network is generated in the following way. When $n=1$, $\mathcal{G}_1$ is a complete graph of 3 vertices. For $n > 1$, each subsequent generation is obtained from the current one by creating a new vertex linked to both endvertices of every existent edge. Figure~\ref{network} shows the networks for the first several generations. By construction,  the number of edges in network $\mathcal{G}_n$ is $E_n=3^{n+1}$.

\begin{figure}
\begin{center}
\includegraphics{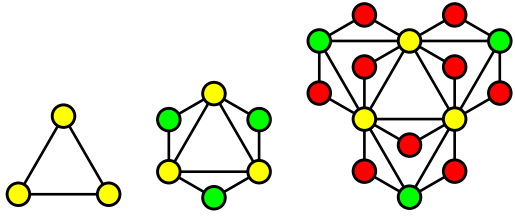} 
\end{center}
\caption[kurzform]{The first three generations of the
scale-free network.} \label{network}
\end{figure}

The resultant network exhibits the prominent properties observed in various real-life systems. First, it is scale-free, with the degree distribution of its vertices obeying a power law form $P(k)\sim k^{1+\ln 3/\ln 2}$~\cite{DoGoMe02}. In addition, it displays the small-world effect, with its average distance increasing logarithmically with the number of vertices~\cite{ZhZhCh07} and the average clustering coefficient converging to a high constant $\frac{4}{5}$.

Another interesting property of the scale-free network is its self-similarity, which is also pervasive in realistic networks~\cite {SoHaMa05}. The initial three vertices of $\mathcal{G}_1$ have the largest degree, which we call hub vertices. For network $\mathcal{G}_n$, we denote its three hub vertices by $A_{n}$, $B_{n}$, and $C_{n}$, respectively. The self-similarity of the network can be seen from its alternative construction method~\cite{ZhZhCh07,ZhLiWuZh10}. Given the $n$th generation network $\mathcal{G}_{n}$, $\mathcal{G}_{n+1}$ can be obtained by merging three replicas of $\mathcal{G}_{n}$ at their hub vertices, see Fig.~\ref{mergeF}. Let $\mathcal{G}_{n}^{\theta}$, $\theta=1,2,3$, be three copies of $\mathcal{G}_{n}$, the hub vertices of which are represented by $A_{n}^{\theta}$, $B_{n}^{\theta}$,
and $C_{n}^{\theta}$, respectively. Then, $\mathcal{G}_{n+1}$ can be obtained by joining $\mathcal{G}_{n}^{\theta}$, with $A_{n}^{1}$
(resp. $C_{n}^{1}$, $A_{n}^{2}$) and $B_{n}^{3}$ (resp. $B_{n}^{2}$,
$C_{n}^{3}$) being identified as the hub vertex $A_{n+1}$ (resp.
$B_{n+1}$, $C_{n+1}$) in $\mathcal{G}_{n+1}$.

\begin{figure}
\begin{center}
\includegraphics[width=.70\linewidth]
{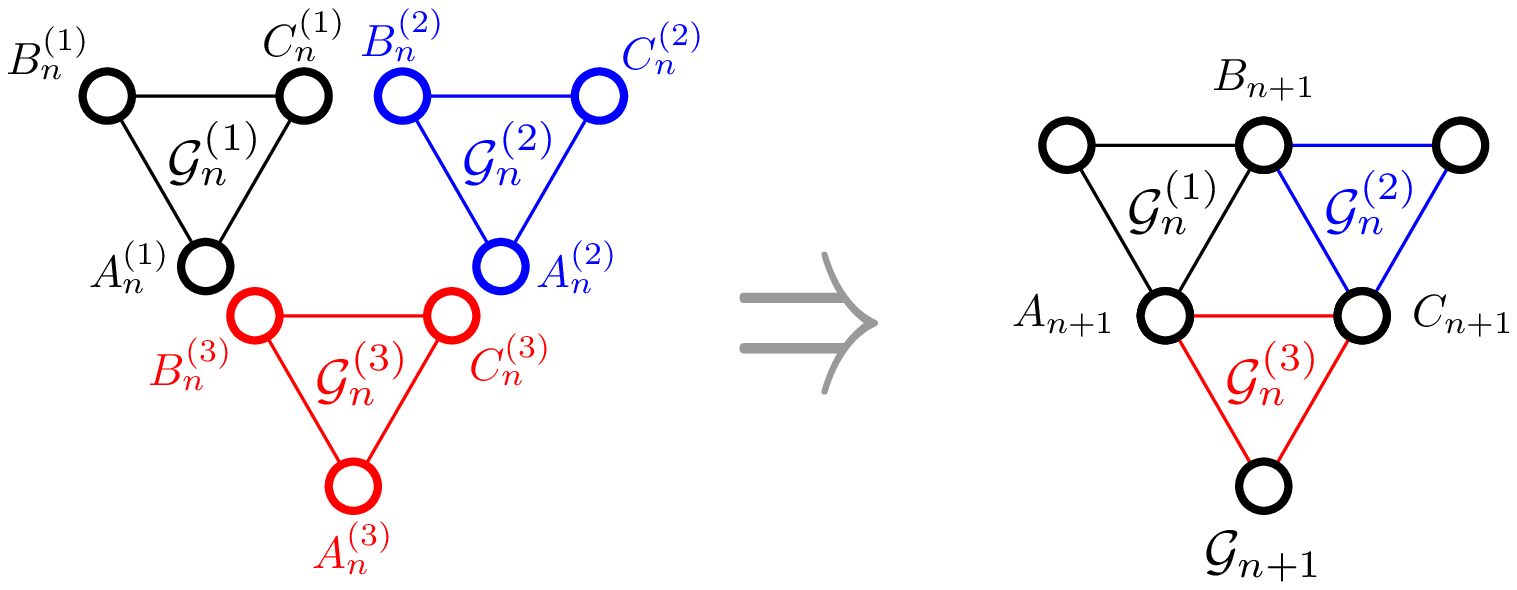}
\caption{Another construction of the scale-free network.} \label{mergeF}
\end{center}
\end{figure}

Let $N_n$ be the number of vertices in $G_{n}$. According to the second construction of the network, $N_n$ obeys the relation $N_{n+1}=3N_n-3$, which together with the initial value $N_0=3$ yields $N_n=(3^{n}+3)/2$.


\subsection{Domination number and minimum dominating set}


\begin{figure}
\begin{center}
\includegraphics[width=0.70\linewidth]
{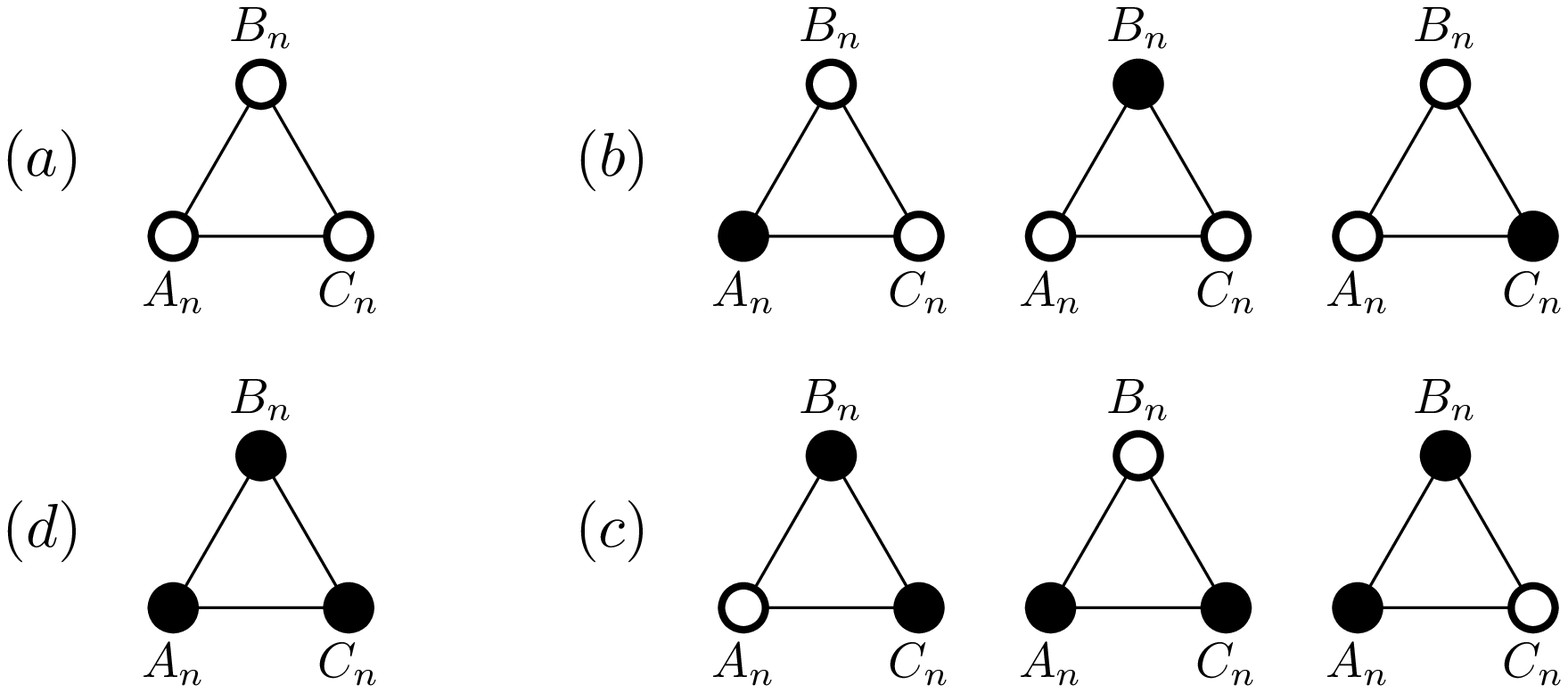}
\end{center}
\caption[kurzform]{\label{Demo01} Illustrations for the definitions of
$\Theta_n^k$, $k = 0,1,2,3$, only showing the three hub vertices. (a), (b), (c) and (d) correspond to a dominating set belonging to $\Theta_n^0$, $\Theta_n^1$, $\Theta_n^2$, and $\Theta_n^3$, respectively. Every filled circle denotes a hub  in the dominating set, while each empty circle represents a hub not in the dominating set.}
\end{figure}

Let $\gamma_n$ denote the domination number of network $\mathcal{G}_n$. In order to determine $\gamma_n$, we define some intermediate quantities. As shown above, there are three hub vertices in $\mathcal{G}_n$, then all dominating sets of $\mathcal{G}_n$ can be sorted into four classes: $\Omega_n^0$, $\Omega_n^1$, $\Omega_n^2$, and $\Omega_n^3$, where $\Omega_n^k$ ($k = 0,1,2,3$) represents those dominating sets, each of which includes exactly $k$ hub vertices. Let $\Theta_n^k$, $k = 0,1,2,3$, be the subset of $\Omega_n^k$, which has the smallest cardinality (number of vertices), denoted by $\gamma_n^k$. Figure~\ref{Demo01} illustrates the definitions for $\Theta_n^k$, $k = 0,1,2,3$.  By definition, we have the following lemma.
\begin{lemma}\label{Dom01}
The domination number of network $\mathcal{G}_n$, $n\geq1$, is $\gamma_n = \min\{\gamma_n^0,\gamma_n^1,\gamma_n^2,\gamma_n^3\}$.
\end{lemma}
After reducing the problem of determining $\gamma_n$ to computing $\gamma_n^k$, $k = 0,1,2,3$, we next evaluate $\gamma_n^k$ by using the self-similar property of the network. 
\begin{lemma}
For two successive generation networks $\mathcal{G}_n$ and $\mathcal{G}_{n+1}$, $n\geq1$,
\begin{equation}\label{Dom02}
\gamma_{n+1}^0 = \min\{3\gamma_n^0,2\gamma_n^0+\gamma_n^1,2\gamma_n^1+\gamma_n^0,3\gamma_n^1\},
\end{equation}
\begin{align}\label{Dom03}
\gamma_{n+1}^1 &= \min\{2\gamma_n^1+\gamma_n^0-1,3\gamma_n^1-1,\gamma_n^0+\gamma_n^1+\gamma_n^2-1,\nonumber\\
&\quad \gamma_n^2+2\gamma_n^1-1,2\gamma_n^2+\gamma_n^0-1,2\gamma_n^2+\gamma_n^1-1\},
\end{align}
\begin{align}\label{Dom04}
\gamma_{n+1}^2 &= \min\{2\gamma_n^1+\gamma_n^2-2,\gamma_n^3+2\gamma_n^1-2,2\gamma_n^2+\gamma_n^1-2,\nonumber\\
&\quad \gamma_n^3+\gamma_n^2+\gamma_n^1-2,3\gamma_n^2-2,\gamma_n^3+2\gamma_n^2-2\},
\end{align}
\begin{equation}\label{Dom05}
\gamma_{n+1}^3 = \min\{3\gamma_n^2-3,\gamma_n^3+2\gamma_n^2-3,2\gamma_n^3+\gamma_n^2-3,3\gamma_n^3-3\}.
\end{equation}
\end{lemma}
\begin{proof}
By definition, $\gamma_{n+1}^k$, $k = 0,1,2,3$, is the cardinality of $\Theta_{n+1}^k$. Below, we will show that the four sets $\Theta_{n+1}^k$, $k = 0,1,2,3$, constitute a complete set, since each one can be constructed iteratively from $\Theta_n^0$, $\Theta_n^1$, $\Theta_n^2$, and $\Theta_n^3$. Then, $\gamma_{n+1}^k$, $k = 0,1,2,3$, can be obtained from $\gamma_n^0$, $\gamma_n^1$, $\gamma_n^2$, and $\gamma_n^3$.

We first consider Eq.~\eqref{Dom02}, which can be proved graphically.

Note that $\mathcal{G}_{n+1}$ is composed of three copies of $\mathcal{G}_n$, $\mathcal{G}_{n}^{\theta}$ ($\theta=1,2,3$). By definition, for any dominating set $\chi$ belonging to $\Theta_{n+1}^0$, the three hub vertices of $\mathcal{G}_{n+1}$ are not in $\chi$, which means that the corresponding six identified hub vertices of $\mathcal{G}_{n}^{\theta}$ do not belong to $\chi$, see Fig.~\ref{mergeF}. Thus, we can construct $\chi$ from $\Theta_n^0$, $\Theta_n^1$, $\Theta_n^2$, and $\Theta_n^3$ by considering whether the other three hub vertices of $\mathcal{G}_{n}^{\theta}$, $\theta=1,2,3$, are in  $\chi$ or not.  Figure~\ref{Theta0} shows all possible configurations of dominating sets $\Omega_{n+1}^0$ that contains $\Theta_{n+1}^0$ as subsets. From Fig.~\ref{Theta0}, we obtain
\begin{equation*}
\gamma_{n+1}^0 = \min\{3\gamma_n^0,2\gamma_n^0+\gamma_n^1,2\gamma_n^1+\gamma_n^0,3\gamma_n^1\}.
\end{equation*}

For Eqs.~\eqref{Dom03},~\eqref{Dom04}, and~\eqref{Dom05}, they can be proved similarly. In Figs.~\ref{Theta1},~\ref{Theta2}, and~\ref{Theta3}, we provide graphical representations of Eqs.~\eqref{Dom03},~\eqref{Dom04}, and~\eqref{Dom05}, respectively.
\end{proof}




\begin{figure}[htbp]
\centering
\includegraphics[width=0.70\linewidth]{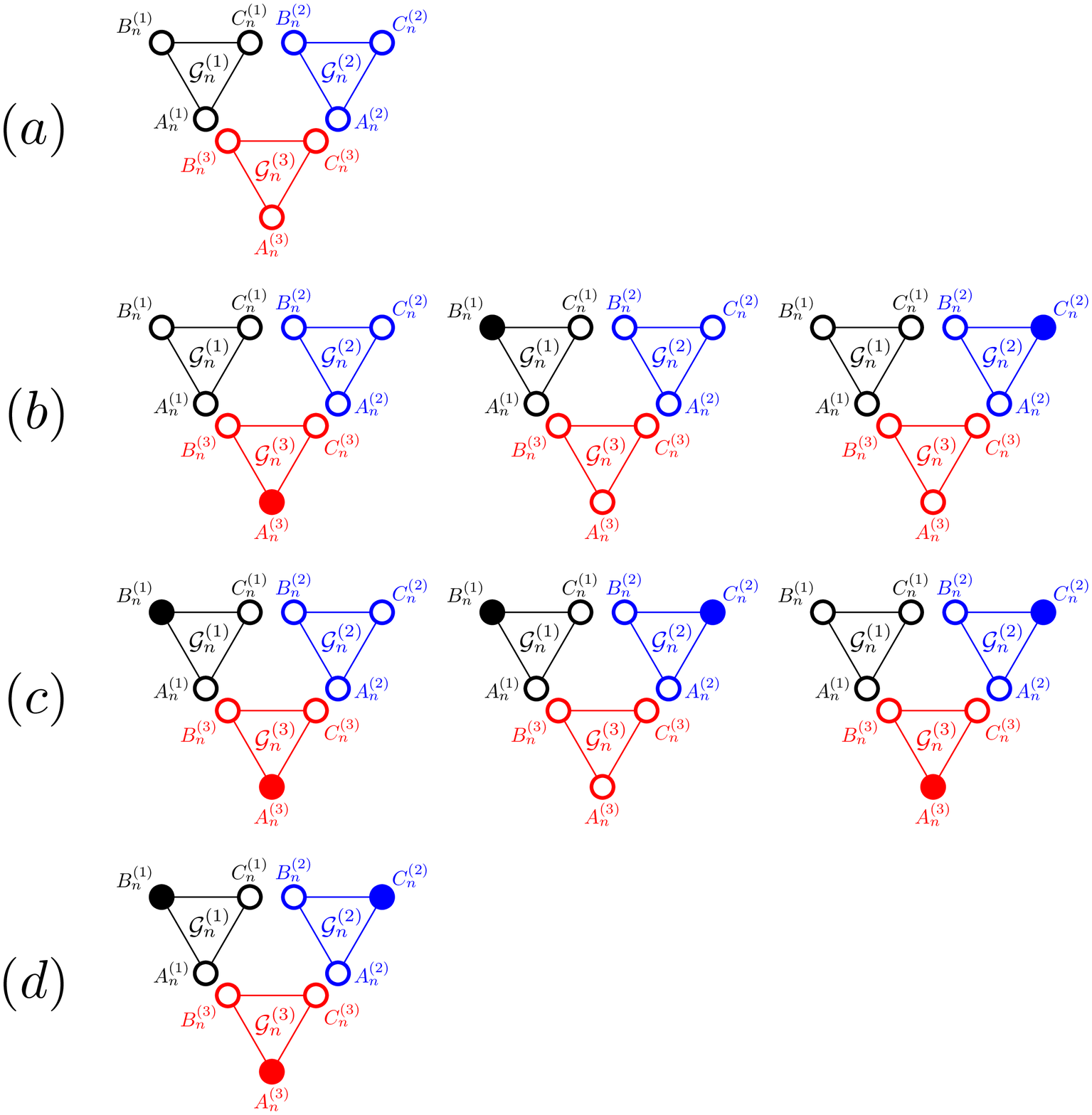}
\caption{\label{Theta0}Illustration of all possible configurations of dominating sets $\Omega_{n+1}^0$ of $\mathcal{G}_{n+1}$, which contain $\Theta_{n+1}^0$. Only the hub vertices of $\mathcal{G}_{n}^{\theta}$, $\theta=1,2,3$, are shown. Solid vertices are in the dominating sets, while open vertices are not.}
\end{figure}

\begin{figure}[htbp]
\centering
\includegraphics[width=1.00\linewidth]{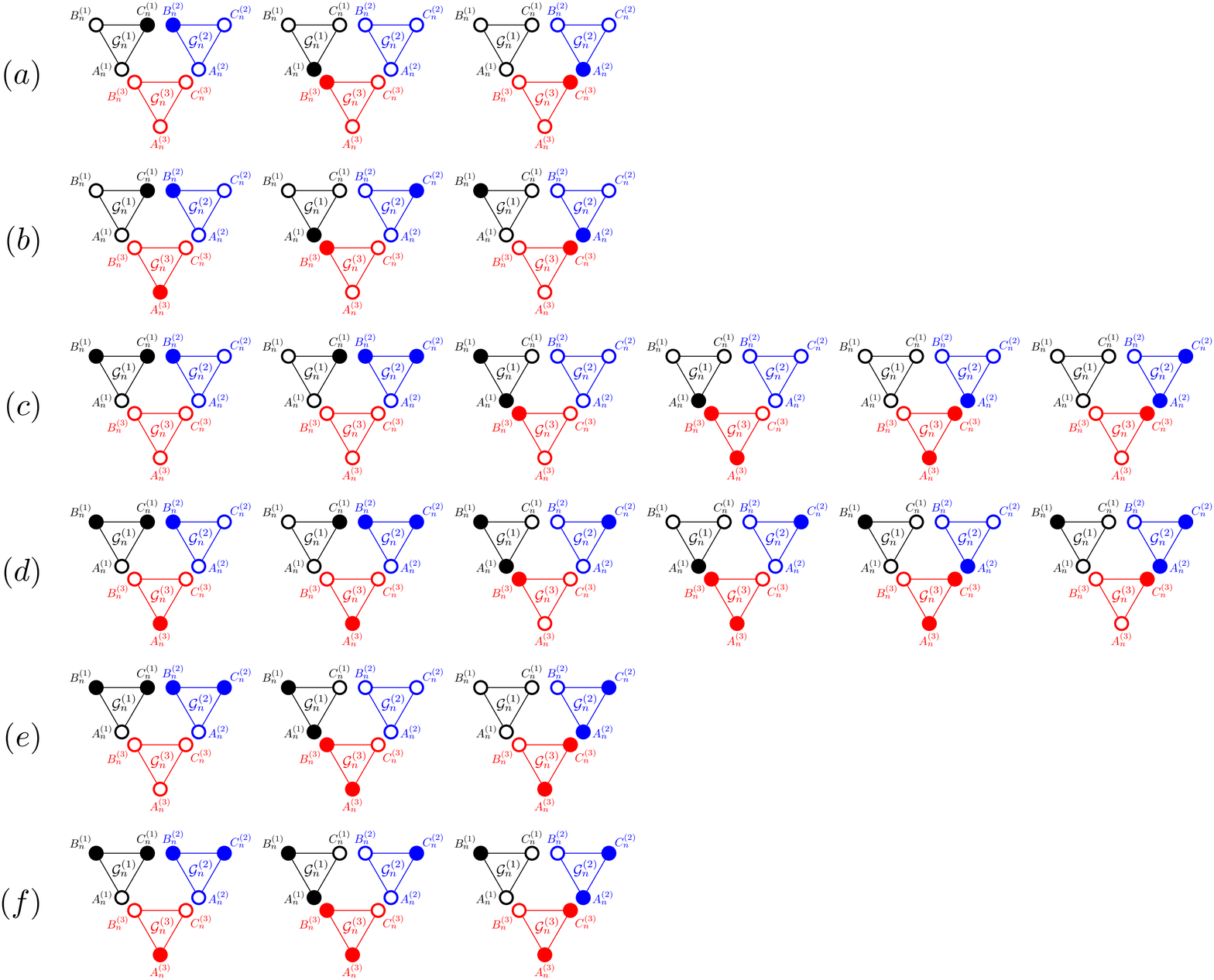}
\caption{\label{Theta1}Illustration of all possible configurations of dominating sets $\Omega_{n+1}^1$ of $\mathcal{G}_{n+1}$, which contain $\Theta_{n+1}^1$.}
\end{figure}

\begin{figure}[htbp]
\centering
\includegraphics[width=1.00\linewidth]{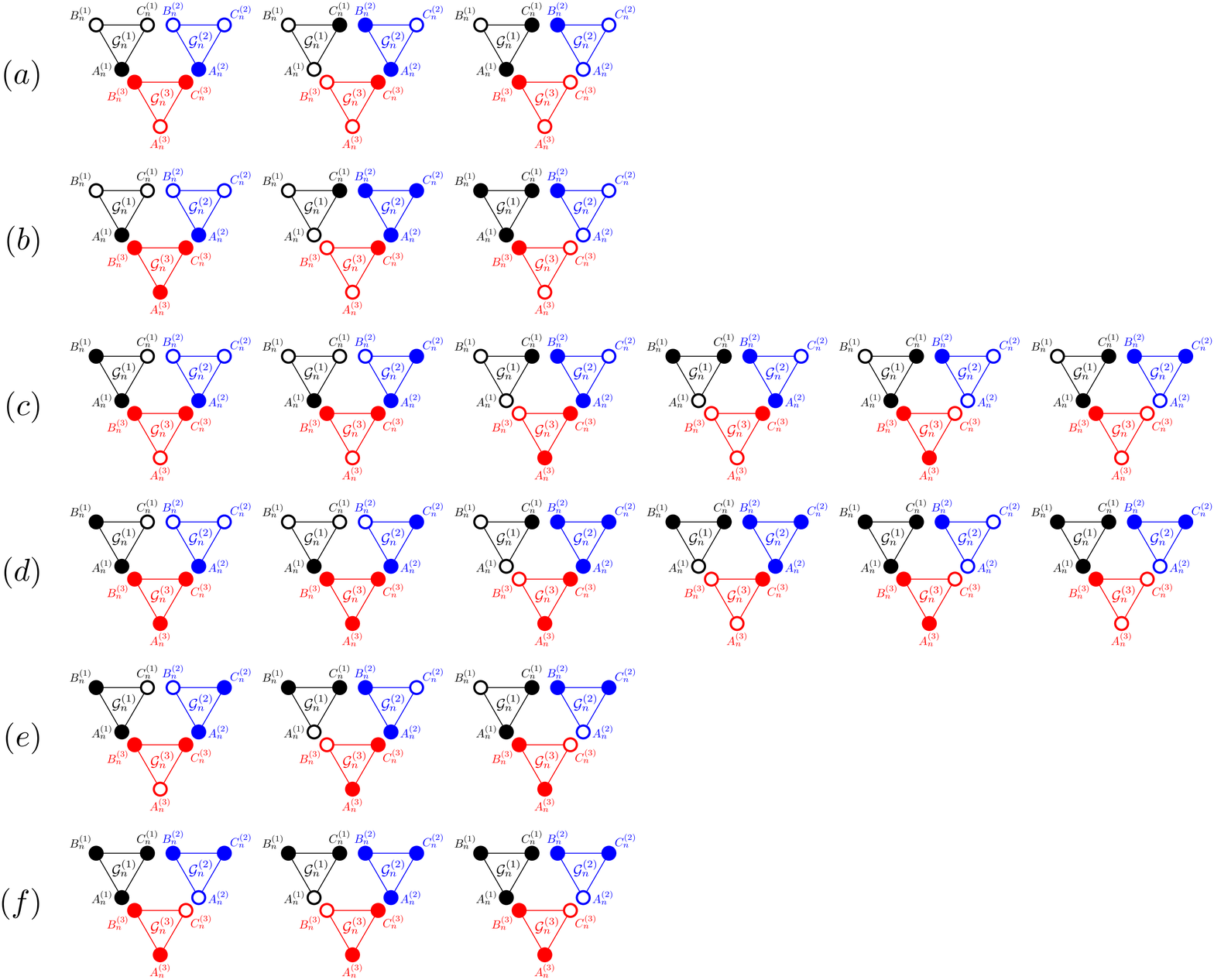}
\caption{\label{Theta2}Illustration of all possible configurations of dominating sets $\Omega_{n+1}^2$ of $\mathcal{G}_{n+1}$, which contain $\Theta_{n+1}^2$.}
\end{figure}

\begin{figure}[htbp]
\centering
\includegraphics[width=0.60\linewidth]{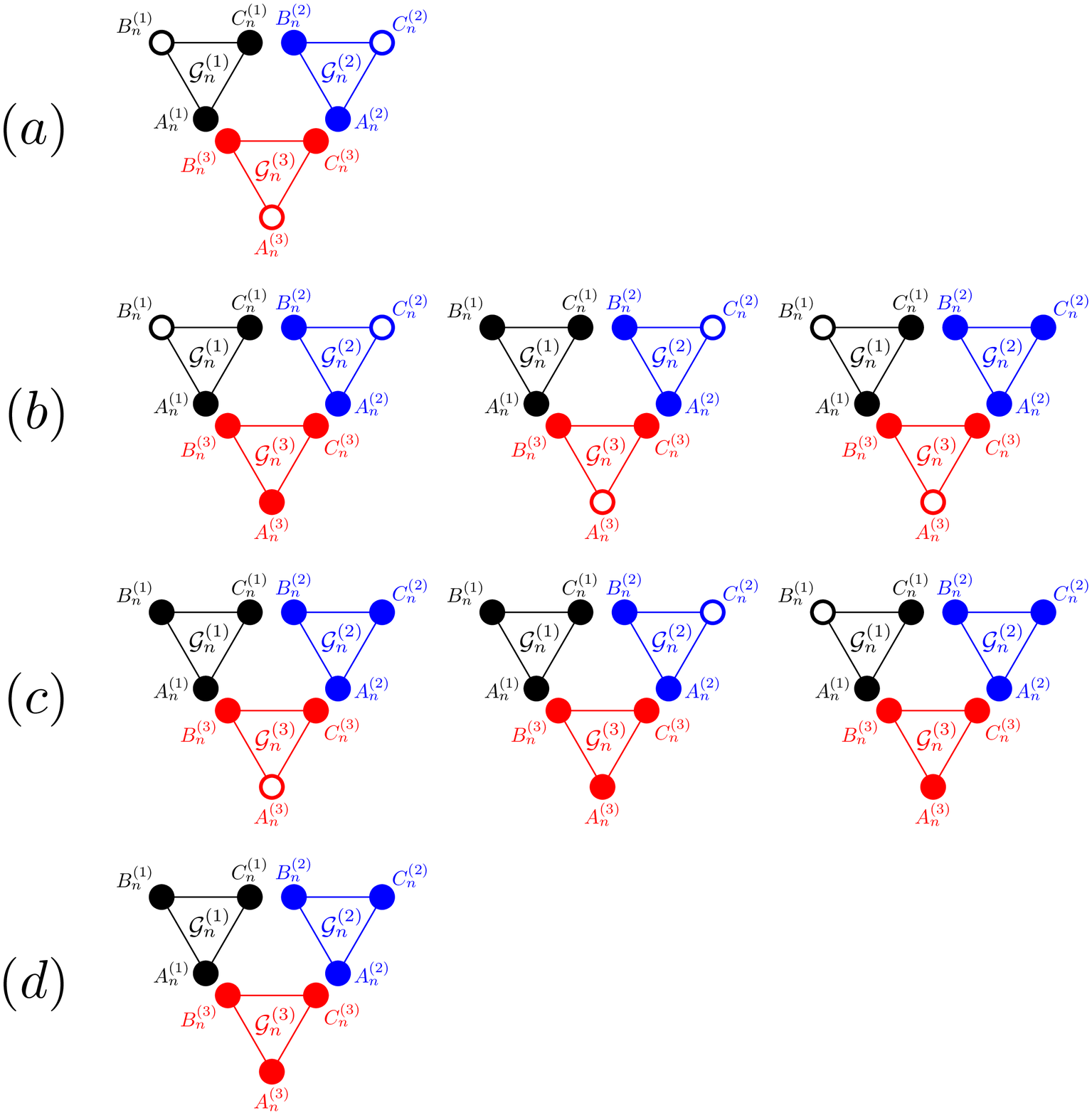}
\caption{\label{Theta3}Illustration of all possible configurations of dominating sets $\Omega_{n+1}^3$ of $\mathcal{G}_{n+1}$, which contain $\Theta_{n+1}^3$.}
\end{figure}

\begin{lemma}\label{Dom06}
For network $\mathcal{G}_n$, $n\geq 3$, $\gamma_n^3 \leq \gamma_n^2 \leq \gamma_n^1 \leq \gamma_n^0$.
\end{lemma}
\begin{proof}
We will prove this lemma by mathematical induction on $n$. For $n=3$, we can obtain $\gamma_3^3=3$, $\gamma_3^2=4$, $\gamma_3^1=5$ and $\gamma_3^0=6$ by hand. Thus, the basis step holds immediately. \par
Assuming that the lemma hold for $n=t$, $t\geq 3$. Then, from Eq.~\eqref{Dom05}, $\gamma_{t+1}^3 = \min\{3\gamma_t^2-3,\gamma_t^3+2\gamma_t^2-3,2\gamma_t^3+\gamma_t^2-3,3\gamma_t^3-3\}$. By induction hypothesis, we have
\begin{equation}\label{Fgamma01}
\gamma_{t+1}^3 = 3\gamma_t^3-3.
\end{equation}
Analogously, we can obtain the following relations:
\begin{equation}\label{Fgamma02}
\gamma_{t+1}^2 = \gamma_t^3 + 2\gamma_t^2-2,
\end{equation}
\begin{equation}\label{Fgamma03}
\gamma_{t+1}^1 = 2\gamma_t^2 + \gamma_t^1 - 1,
\end{equation}
\begin{equation}\label{Fgamma04}
\gamma_{t+1}^0 = 3\gamma_t^1.
\end{equation}
By comparing the above Eqs.~(\ref{Fgamma01}-\ref{Fgamma04}) and using the induction hypothesis $\gamma_t^3 \leq \gamma_t^2 \leq \gamma_t^1 \leq \gamma_t^0$, we have $\gamma_{t+1}^3 \leq \gamma_{t+1}^2 \leq \gamma_{t+1}^1 \leq \gamma_{t+1}^0$. Therefore, the lemma is true for $n=t+1$.

This concludes the proof of the lemma.
\end{proof}

\begin{theorem}\label{SFdomN}
The domination number of network $\mathcal{G}_n$, $n\geq3$, is
\begin{equation}\label{Fgamma04x}
\gamma_n = \frac{3^{n-2}+3}{2}\,.
\end{equation}
\end{theorem}
\begin{proof}
By Lemma~\ref{Dom01} and Eq.~\eqref{Fgamma01}, we obtain
\begin{equation}\label{Fgamma05}
\gamma_{n+1} = \gamma_{n+1}^3 = 3\gamma_n^3 - 3 = 3\gamma_n - 3\,,
\end{equation}
which, with the initial condition $\gamma_{3}=3$, is solved to yield the result.
\end{proof}

Theorem~\ref{SFdomN}, especially Eq.~\eqref{Fgamma05}, implies that any minimum dominating set of $\mathcal{G}_n$ includes the three hub vertices.
\begin{lemma}
The smallest number of vertices in a dominating set of $\mathcal{G}_n$, $n\geq3$, which contains exactly $2$, $1$, and $0$ hub vertices, is
\begin{equation}\label{Fgamma06}
\gamma_n^2 = \frac{3^{n-2}+1}{2}+2^{n-2},
\end{equation}
\begin{equation}\label{Fgamma07}
\gamma_n^1 = \frac{3^{n-2}-1}{2}+2^{n-1}
\end{equation}
and
\begin{equation}\label{Fgamma08}
\gamma_n^0 = \frac{3^{n-2}-3}{2}+3\cdot 2^{n-2},
\end{equation}
respectively.
\end{lemma}
\begin{proof}
By Eqs.~\eqref{Fgamma04x} and~\eqref{Fgamma05}, we have $\gamma_n^3=\gamma_n=\frac{3^{n-2}+3}{2}$. Considering  Eq.~\eqref{Fgamma02}, we obtain a recursive equation for $\gamma_n^2$ as:
\begin{equation}
\gamma_{n+1}^2 = 2\gamma_n^2 + \frac{3^{n-2}+3}{2}-2,
\end{equation}
which coupled with $\gamma_3^2 = 4$ is solved to give Eq.~\eqref{Fgamma06}.

Analogously, we can prove Eqs.~\eqref{Fgamma07} and~\eqref{Fgamma08}.
\end{proof}

\begin{theorem}
For network $\mathcal{G}_n$,  $n \geq 3$, there is a unique minimum dominating set.
\end{theorem}

\begin{proof}
Equation~\eqref{Fgamma05} and Fig.~\ref{Theta3} shows that for $n \geq 3$ any minimum dominating set of $\mathcal{G}_{n+1}$ is in fact the union of minimum dominating sets, $\Theta_n^3$, of the three replicas of $\mathcal{G}_{n}$ (i.e. $\mathcal{G}_{n}^{1}$, $\mathcal{G}_{n}^{2}$, and $\mathcal{G}_{n}^{3}$) forming $\mathcal{G}_{n+1}$, with each pair of their identified hub vertices being counted only once. Thus, any minimum dominating set of $\mathcal{G}_{n+1}$ is determined by those of $\mathcal{G}_{n}^{1}$, $\mathcal{G}_{n}^{2}$, and $\mathcal{G}_{n}^{3}$.  Since the minimum dominating set of $\mathcal{G}_3$ is unique, there is unique  minimum dominating set for $G_n$ when $n \geq 3$. Moreover, it is easy to see that the unique dominating set of $\mathcal{G}_n$, $n \geq 3$, is actually the set of all vertices of $\mathcal{G}_{n-2}$.
\end{proof}

\section{Domination number and minimum dominating sets in Sierpi\'nski graph}

In this section, we study the domination number and the number of minimum dominating sets in the Sierpi\'nski graph, and compare the results with those of the above-studied scale-free network, aiming to uncover the effect of scale-free property on the domination number and the number minimum dominating sets.

\subsection{Construction of Sierpi\'nski graph}

The Sierpinski graph is also built in an iterative way. Let  $\mathcal{S}_n$, $n\geq 1$, denote the $n$-generation graph. Then the Sierpi\'nski graph is created as follows. When $n=1$, $\mathcal{S}_1$ is an equilateral triangle containing three vertices and three edges. For $n =2$, perform a bisection of the sides of  $\mathcal{S}_1$   forming four small replicas of the original equilateral triangle, and remove the central downward pointing triangle to get $\mathcal{S}_2$. For $n>2$, $\mathcal{S}_n$ is obtained from $\mathcal{S}_{n-1}$ by performing the bisecting and removing operations for each triangle in $\mathcal{S}_{n-1}$. Figure~\ref{Sierp01} shows theSierpi\'nski graphs for $n=1,2,3$.

\begin{figure}
\begin{center}
\includegraphics[width=0.60\textwidth]{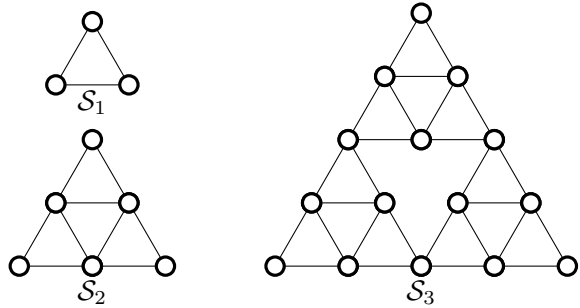} 
\end{center}
\caption[kurzform]{The first three generations of the
Sierpi\'nski graph.} \label{Sierp01}
\end{figure}

It is easy to verify that both the number of vertices and the number of edges in the Sierpi\'nski $\mathcal{S}_n$ graph are identical to those for the scale-free network $\mathcal{G}_n$, which are $N_n=(3^{n}+3)/2$ and $E_n=3^{n+1}$, respectively.

Different from $\mathcal{G}_n$, the Sierpi\'nski graph is homogeneous, since the degree of vertices in $\mathcal{S}_n$ is equal to 3, excluding the topmost vertex $A_n$, leftmost vertex $B_n$ and the rightmost vertex $C_n$, the degree of which is 2. We call these three vertices with degree 2 as outmost vertices.

\begin{figure}
\begin{center}
\includegraphics[width=8.5cm]{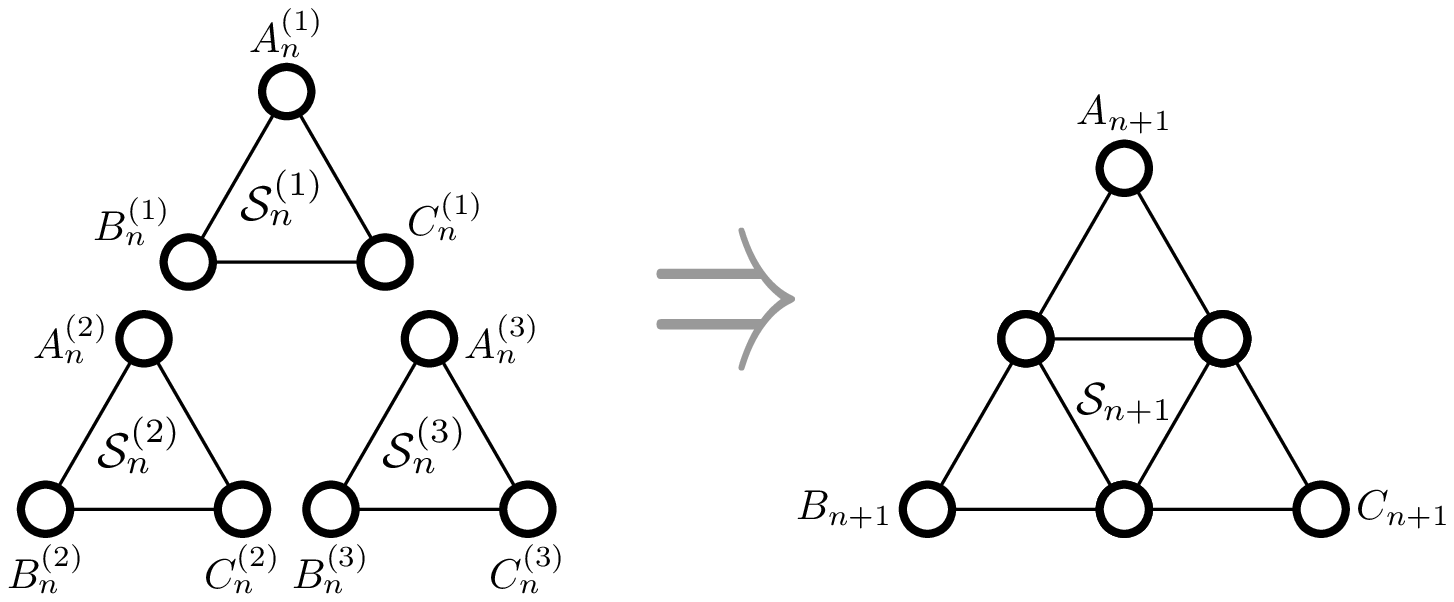}
\caption{Alternative construction of the Sierpi\'nski graph.} \label{merge}
\end{center}
\end{figure}

The Sierpi\'nski graph is also self-similar, which suggests another construction method of the graph. Given the $n$th generation graph $\mathcal{S}_{n}$, the $(n+1)$th generation graph $\mathcal{S}_{n+1}$ can be obtained by joining three replicas of $\mathcal{S}_{n}$ at their outmost vertices, see Fig.~\ref{merge}. Let $\mathcal{S}_{n}^{\theta}$, $\theta=1,2,3$, be three copies of $\mathcal{S}_{n}$, the outmost vertices of which are represented by $A_{n}^{\theta}$, $B_{n}^{\theta}$, and $C_{n}^{\theta}$, respectively. Then, $\mathcal{S}_{n+1}$ can be obtained by joining $\mathcal{S}_{n}^{\theta}$, with $A_{n}^{1}$, $B_{n}^{2}$, and $C_{n}^{3}$ being the outmost vertices $A_{n+1}$, $B_{n+1}$, and $C_{n+1}$  of $\mathcal{S}_{n+1}$.


\subsection{Domination number}

In the case without confusion, for the Sierpi\'nski graph $\mathcal{S}_n$ we employ the same notation as those for $\mathcal{G}_n$ studied in the preceding section. Let $\gamma_n$ be the domination number of $\mathcal{S}_n$. Note that all dominating sets of $\mathcal{S}_n$ can be classified into four types: $\Omega_n^0$, $\Omega_n^1$, $\Omega_n^2$, and $\Omega_n^3$, where $\Omega_n^k$ ($k = 0,1,2,3$) denotes those dominating sets, each including exactly $k$ outmost vertices. Let $\Theta_n^k$, $k = 0,1,2,3$, be the subsets of $\Omega_n^k$, each of which has the smallest cardinality, denoted by $\gamma_n^k$. 


\begin{lemma}\label{leSGDom01}
The domination number of the Sierpi\'nski graph $\mathcal{S}_n$, $n\geq1$, is $\gamma_n = \min\{\gamma_n^0,\gamma_n^1,\gamma_n^2,\gamma_n^3\}$.
\end{lemma}
Thus, in order to  determine $\gamma_n$ for $\mathcal{S}_n$, we can alternatively evaluate $\gamma_n^k$, $k = 0,1,2,3$, which can be solved by using the self-similar structure of the Sierpinski graph.

To determine $\gamma_n^k$, $k = 0,1,2,3$, we introduce some more quantities assisting the calculation. We use $\mathcal{S}_n^k$ to denote a subgraph of the Sierpi\'nski graph $\mathcal{S}_n$, which is obtained from $\mathcal{S}_n$ by removing $k$, $k =1,2,3$, outmost vertices and the edges incident to them. Let $\phi_n^i$, $i=0,1,2$ be the smallest number of vertices in a dominating set of $\mathcal{S}_n^1$ containing $i$ outmost vertices but excluding either of the two neighbors of the removed outmost vertex; let  $\xi_n^j$, $j=0,1$, be the smallest number of vertices in a dominating set of $\mathcal{S}_n^2$ including $j$ outmost vertex but excluding any neighbor of the two removed outmost vertices; and let $\eta_n$ denote the smallest number of vertices in a dominating set of $\mathcal{S}_n^3$, which does not include the neighbors of the three removed outmost vertices.

According to the self-similar property of the Sierpi\'nski graph, we can establish some relations between the quantities defined above.
\begin{lemma}
\label{leSGDom02}
For any integer $n \geq 3$, the following relations hold.
\begin{small}
\begin{align}\label{SGDom01}
\gamma_{n+1}^0 &= \min\{\gamma_n^0+\phi_n^0+\xi_n^0, 3\phi_n^0, 2\gamma_n^0 +\xi_n^0, \gamma_n^0+2\phi_n^0, 2\gamma_n^0+\phi_n^0, 3\gamma_n^0, \nonumber\\
&\quad 2\gamma_n^1+\xi_n^0-1, \gamma_n^0+2\phi_n^1-1, \gamma_n^1+\phi_n^1+\phi_n^0-1, 2\gamma_n^1+\phi_n^0-1, \gamma_n^1+\gamma_n^0+\phi_n^1-1,\nonumber\\
&\quad 2\gamma_n^1+\gamma_n^0-1, \gamma_n^2+\gamma_n^1+\phi_n^1-2, \gamma_n^2+2\gamma_n^1-2,3\gamma_n^2-3\},
\end{align}
\begin{align}\label{SGDom02}
\gamma_{n+1}^1 &= \min\{\gamma_n^0+\phi_n^1+\xi_n^0, \gamma_n^1+\phi_n^0+\xi_n^0, \gamma_n^0+\phi_n^0+\xi_n^1, \phi_n^1+2\phi_n^0, 2\gamma_n^0+\xi_n^1, \gamma_n^0+\phi_n^1+\phi_n^0, \nonumber\\
&\quad \gamma_n^1+\gamma_n^0+\xi_n^0, \gamma_n^1+2\phi_n^0, 2\gamma_n^0+\phi_n^1, \gamma_n^1+\gamma_n^0+\phi_n^0, \gamma_n^1+2\gamma_n^0, \nonumber\\
&\quad \gamma_n^2+\gamma_n^1+\xi_n^0-1, \gamma_n^0+\phi_n^2+\phi_n^1-1, \gamma_n^1+\phi_n^2+\phi_n^0-1,\gamma_n^2+\phi_n^1+\phi_n^0-1, \nonumber\\
&\quad \gamma_n^1+\gamma_n^0+\phi_n^2-1,\gamma_n^2+\gamma_n^0+\phi_n^1-1,\gamma_n^2+\gamma_n^1+\phi_n^0-1,\gamma_n^2+\gamma_n^1+\gamma_n^0-1,\nonumber\\
&\quad2\gamma_n^1+\xi_n^1-1, \gamma_n^1+2\phi_n^1-1, 2\gamma_n^1+\phi_n^1-1, 3\gamma_n^1-1,\gamma_n^3+\gamma_n^1+\phi_n^1-2,\nonumber\\
&\quad \gamma_n^3+2\gamma_n^1-2, \gamma_n^2+\gamma_n^1+\phi_n^2-2,2\gamma_n^2+\phi_n^1-2, 2\gamma_n^2+\gamma_n^1-2,\gamma_n^3+2\gamma_n^2-3\},
\end{align}
\begin{align}\label{SGDom03}
\gamma_{n+1}^2 &= \min\{\gamma_n^0+\phi_n^1+\xi_n^1, \gamma_n^1+\phi_n^0+\xi_n^1,\gamma_n^1+\phi_n^1+\xi_n^0, 2\phi_n^1+\phi_n^0, 2\gamma_n^1+\xi_n^0, \gamma_n^1+\gamma_n^0+\xi_n^1,\nonumber\\
&\quad \gamma_n^0+2\phi_n^1, \gamma_n^1+\phi_n^1+\phi_n^0, 2\gamma_n^1+\phi_n^0, \gamma_n^1+\gamma_n^0+\phi_n^1, 2\gamma_n^1+\gamma_n^0, \nonumber\\
&\quad \gamma_n^0+2\phi_n^2-1, \gamma_n^2+\phi_n^2+\phi_n^0-1, 2\gamma_n^2+\xi_n^0-1, \gamma_n^2+\gamma_n^0+\phi_n^2-1,2\gamma_n^2+\phi_n^0-1,\nonumber\\
&\quad 2\gamma_n^2+\gamma_n^0-1, \gamma_n^1+\phi_n^2+\phi_n^1-1, \gamma_n^2+\gamma_n^1+\xi_n^1-1,\gamma_n^2+2\phi_n^1-1,\nonumber\\
&\quad\gamma_n^2+\gamma_n^1+\phi_n^1-1, 2\gamma_n^1+\phi_n^2-1, \gamma_n^2+2\gamma_n^1-1,\gamma_n^3+\gamma_n^2+\phi_n^1-2, \nonumber\\
&\quad\gamma_n^3+\gamma_n^1+\phi_n^2-2,\gamma_n^3+\gamma_n^2+\gamma_n^1-2,2\gamma_n^2+\phi_n^2-2,3\gamma_n^2-2,2\gamma_n^3+\gamma_n^2-3\},
\end{align}
\begin{align}\label{SGDom04}
\gamma_{n+1}^3 &= \min\{\gamma_n^1+\phi_n^1+\xi_n^1, 3\phi_n^1, 2\gamma_n^1+\xi_n^1, \gamma_n^1+2\phi_n^1,2\gamma_n^1+\phi_n^1,3\gamma_n^1,\nonumber\\
&\quad 2\gamma_n^2+\xi_n^1-1,\gamma_n^1+2\phi_n^2-1,\gamma_n^2+\phi_n^2+\phi_n^1-1,2\gamma_n^2+\phi_n^1-1,\gamma_n^2+\gamma_n^1+\phi_n^2-1,\nonumber\\
&\quad 2\gamma_n^2+\gamma_n^1-1,\gamma_n^3+\gamma_n^2+\phi_n^2-2,\gamma_n^3+2\gamma_n^2-2,3\gamma_n^3-3\},
\end{align}
\begin{align}\label{SGDom05}
\phi_{n+1}^0 &= \min\{\gamma_n^0+2\xi_n^0, 2\phi_n^0+\xi_n^0, \gamma_n^0+\phi_n^0+\eta_n, 2\gamma_n^0+\eta_n, \gamma_n^0+\phi_n^0+\xi_n^0,\nonumber\\
&\quad 3\phi_n^0, 2\gamma_n^0+\xi_n^0, \gamma_n^0+2\phi_n^0, 2\gamma_n^0+\phi_n^0, \gamma_n^1+\phi_n^1+\xi_n^0-1, 2\phi_n^1+\phi_n^0-1,  \nonumber\\
&\quad \gamma_n^1+\phi_n^0+\xi_n^1-1, \gamma_n^0+\phi_n^1+\xi_n^1-1, \gamma_n^0+2\phi_n^1-1, \gamma_n^1+\gamma_n^0+\xi_n^1-1,\nonumber\\
&\quad \gamma_n^1+\phi_n^1+\phi_n^0-1, \gamma_n^1+\gamma_n^0+\phi_n^1-1, 2\gamma_n^1+\eta_n-1, 2\phi_n^1+\phi_n^0-1, \nonumber\\
&\quad \gamma_n^1+\phi_n^1+\xi_n^0-1, 2\gamma_n^1+\xi_n^0-1, \gamma_n^1+\pi_n^1+\phi_n^0-1, 2\gamma_n^1+\phi_n^0-1, \nonumber\\
&\quad \gamma_n^1+\phi_n^2+\phi_n^1-2, 2\gamma_n^1+\phi_n^2-2, \gamma_n^2+\gamma_n^1+\xi_n^1-2, \gamma_n^2+2\phi_n^1-2, \nonumber\\
&\quad \gamma_n^2+\gamma_n^1+\phi_n^1-2, 2\gamma_n^2+\phi_n^2-3\},
\end{align}
\begin{align}\label{SGDom06}
\phi_{n+1}^1 &= \min\{\gamma_n^1+2\xi_n^0, 2\phi_n^0+\xi_n^1, \gamma_n^0+\xi_n^1+\xi_n^0, \phi_n^1+\phi_n^0
+\xi_n^0, \gamma_n^1+\phi_n^0+\eta_n, \gamma_n^0+\phi_n^1+\eta_n, \nonumber\\
&\quad \phi_n^1+\phi_n^0+\xi_n^0, \phi_n^1+2\phi_n^0, \gamma_n^0+\phi_n^1+\xi_n^0, \gamma_n^1+\phi_n^0+\xi_n^0, \gamma_n^0+\phi_n^1+\phi_n^0, \gamma_n^0+\gamma_n^1+\xi_n^0,\nonumber\\
&\quad \gamma_n^1+2\phi_n^0, \gamma_n^1+\gamma_n^0+\phi_n^0, \gamma_n^2+\phi_n^1+\xi_n^0-1, \gamma_n^0+\phi_n^2+\xi_n^1-1, \gamma_n^2+\phi_n^0+\xi_n^1-1, \nonumber\\
&\quad \phi_n^2+\phi_n^1+\phi_n^0-1, \gamma_n^0+\phi_n^2+\phi_n^1-1, \gamma_n^2+\phi_n^1+\phi^0-1, \gamma_n^2+\gamma_n^0+\xi_n^1-1, \nonumber\\
&\quad \gamma_n^2+\gamma_n^0+\phi_n^1-1, \gamma_n^1+\phi_n^1+\xi_n^1-1, 3\phi_n^1-1, \gamma_n^1+2\phi_n^1-1, 2\gamma_n^1+\xi_n^1-1,2\gamma_n^1+\phi_n^1-1, \nonumber\\
&\quad \gamma_n^2+\gamma_n^1+\eta_n-1, \gamma_n^1+\phi_n^2+\xi_n^0-1, \gamma_n^2+\gamma_n^1+\xi_n^0-1, \gamma_n^1+\phi_n^2+\phi_n^0-1, \nonumber\\
&\quad\gamma_n^2+\gamma_n^1+\phi_n^0-1, \gamma_n^2+\phi_n^2+\phi_n^1-2,  \gamma_n^1+2\phi_n^2-2, \gamma_n^2+\gamma_n^1+\phi_n^2-2,2\gamma_n^2+\xi_n^1-2,  \nonumber\\
&\quad 2\gamma_n^2+\phi_n^1-2, \gamma_n^3+\gamma_n^1+\xi_n^1-2, \gamma_n^3+2\phi_n^1-2, \gamma_n^3+\gamma_n^1+\phi_n^1-2,\gamma_n^3+\gamma_n^2+\phi_n^2-3\},
\end{align}
\begin{align}\label{SGDom07}
\phi_{n+1}^2 &= \min\{\phi_n^1+\phi_n^0+\xi_n^1, \gamma_n^1+\phi_n^1+\eta_n, 2\phi_n^1+\xi_n^0, \gamma_n^1+\xi_n^1+\xi_n^0, 2\phi_n^1+\phi_n^0, \gamma_n^1+\phi_n^1+\xi_n^0,\nonumber\\
&\quad 2\gamma_n^1+\xi_n^0, \gamma_n^1+\phi_n^1+\phi_n^0, 2\gamma_n^1+\phi_n^0, 2\phi_n^2+\phi_n^0-1, 2\gamma_n^2+\eta_n-1, \gamma_n^2+\phi_n^2+\xi_n^0-1,\nonumber\\
&\quad \gamma_n^2+\phi_n^2+\phi_n^0-1, 2\gamma_n^2+\xi_n^0-1, 2\gamma_n^2+\phi_n^0-1,\gamma_n^1+\phi_n^2+\xi_n^1-1,\gamma_n^2+\phi_n^1+\xi_n^1-1, \nonumber\\
&\quad \phi_n^2+2\phi_n^1-1, \gamma_n^2+\gamma_n^1+\xi_n^1-1, \gamma_n^2+2\phi_n^1-1, \gamma_n^1+\phi_n^2+\phi_n^1-1,\gamma_n^2+\gamma_n^1+\phi_n^1-1,\nonumber\\
&\quad \gamma_n^2+2\phi_n^2-2, 2\gamma_n^2+\phi_n^2-2, \gamma_n^3+\gamma_n^2+\xi_n^1-2, \gamma_n^3+\phi_n^2+\phi_n^1-2, \gamma_n^3+\gamma_n^2+\phi_n^1-2, \nonumber\\
&\quad 2\gamma_n^3+\phi_n^2-3\},
\end{align}
\begin{align}\label{SGDom08}
\xi_{n+1}^0 &= \min\{\gamma_n^0+\xi_n^0+\eta_n, \phi_n^0+2\xi_n^0, 2\phi_n^0+\eta_n, \gamma_n^0+\phi_n^0+\eta_n, 2\phi_n^0+\xi_n^0, \gamma_n^0+2\xi_n^0, 3\phi_n^0, \nonumber\\
&\quad\gamma_n^0+\phi_n^0+\xi_n^0,\gamma_n^0+2\phi_n^0, \gamma_n^0+2\xi_n^1-1, 2\phi_n^1+\xi_n^0-1, \phi_n^1+\phi_n^0+\xi_n^1-1, 2\phi_n^1+\phi_n^0-1, \nonumber\\
&\quad \gamma_n^0+\phi_n^1+\xi_n^1-1, \gamma_n^0+2\phi_n^1-1,\gamma_n^1+\phi_n^1+\eta_n-1, \gamma_n^1+\xi_n^1+\xi_n^0-1, \gamma_n^1+\phi_n^1+\xi_n^0-1, \nonumber\\
&\quad \gamma_n^1+\phi_n^0+\xi_n^1-1, \gamma_n^1+\xi_n^1+\xi_n^0-1, \gamma_n^2+\phi_n^1+\xi_n^1-2, \gamma_n^2+2\phi_n^1-2, \nonumber\\
&\quad \gamma_n^1+\phi_n^2+\xi_n^1-2, \gamma_n^1+\phi_n^2+\phi_n^1-2, \phi_n^2+2\phi_n^1-2, \gamma_n^2+2\phi_n^2-3\},
\end{align}
\begin{align}\label{SGDom09}
\xi_{n+1}^1 &= \min\{\gamma_n^1+\xi_n^0+\eta_n, \phi_n^0+\xi_n^1+\xi_n^0, \phi_n^1+\phi_n^0+\eta_n, \phi_n^1+2\xi_n^0, \phi_n^1+\phi_n^0+\xi_n^0, \gamma_n^1+2\xi_n^0,\nonumber\\
&\quad\phi_n^1+2\phi_n^0, \gamma_n^1+\phi_n^0+\xi_n^0, \gamma_n^1+2\phi_n^0, \gamma_n^1+2\xi_n^1-1, 2\phi_n^1+\xi_n^1-1, 3\phi_n^1-1, \nonumber\\
&\quad \gamma_n^1+\phi_n^1+\xi_n^1-1, \gamma_n^1+2\phi_n^1-1, \phi_n^2+\phi_n^0+\xi_n^1-1, \gamma_n^2+\phi_n^1+\eta_n-1, \gamma_n^2+\xi_n^1+\xi_n^0-1, \nonumber\\
&\quad\gamma_n^2+\phi_n^1+\xi_n^0 -1, \gamma_n^2+\phi_n^1+\xi_n^0-1, \phi_n^2+\phi_n^1+\phi_n^0-1, \gamma_n^2+\phi_n^0+\xi_n^1-1, \nonumber\\
&\quad \gamma_n^2+\phi_n^1+\phi_n^0-1, \gamma_n^3+\phi_n^1+\xi_n^1-2, \gamma_n^3+2\phi_n^1-2, \gamma_n^2+\phi_n^2+\xi_n^1-2, \nonumber\\
&\quad 2\phi_n^2+\phi_n^1-2, \gamma_n^2+\phi_n^2+\phi_n^1-2, \gamma_n^3+2\phi_n^2-3\},
\end{align}
\begin{align}\label{SGDom10}
\eta_{n+1} &= \min\{\phi_n^0+\xi_n^0+\eta_n, 3\xi_n^0, \phi_n^0+2\xi_n^0, 2\phi_n^0+\eta_n, 2\phi_n^0+\xi_n^0,  3\phi_n^0, \phi_n^0+2\xi_n^1-1, 2\phi_n^1+\eta_n-1, \nonumber\\
&\quad \phi_n^1+\xi_n^1+\xi_n^0-1, 2\phi_n^1+\xi_n^0-1, \phi_n^1+\phi_n^0+\xi_n^1-1, 2\phi_n^1+\phi_n^0-1, \phi_n^2+\phi_n^1+\xi_n^1-2, \nonumber\\
&\quad\phi_n^2+2\phi_n^1-2, 3\phi_n^2-3\}.
\end{align}
\end{small}
\end{lemma}

\begin{proof}
This lemma can be proved graphically. Figs.~\ref{SGTheta0}-\ref{SGeta0} illustrate the graphical representations from Eq.~\eqref{SGDom01} to Eq.~\eqref{SGDom10}.
\end{proof}

\begin{figure}[htbp]
\centering
\includegraphics[width=\figwidth]{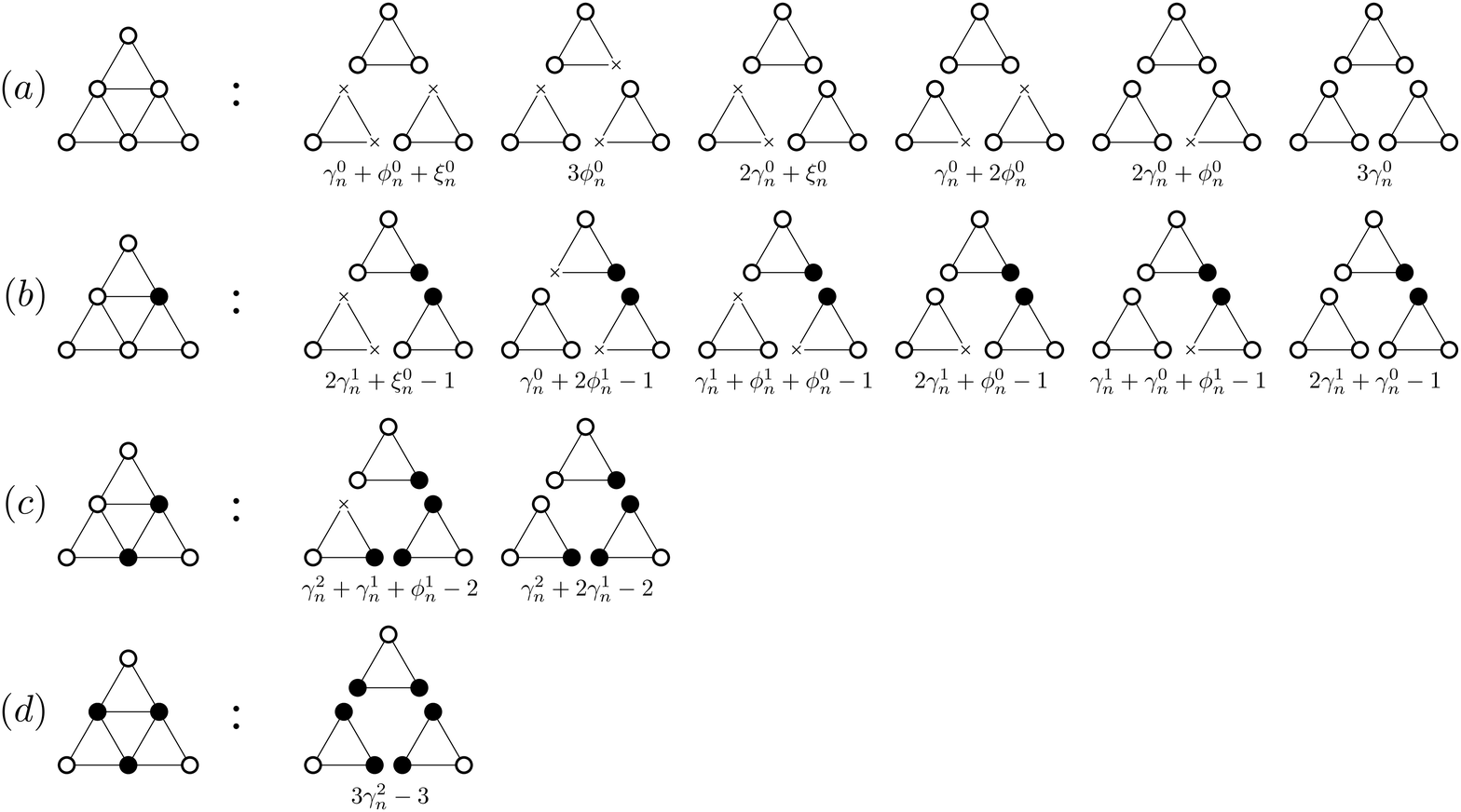}
\caption{\label{SGTheta0}Illustration of all possible configurations of dominating sets $\Omega_{n+1}^0$ of $\mathcal{S}_{n+1}$, which contain $\Theta_{n+1}^0$. Only the outmost vertices of $\mathcal{S}_{n}^{\theta}$, $\theta=1,2,3$, are shown. Solid vertices are in the dominating sets, open vertices are not, while cross vertices denote those removed outmost vertices in $\mathcal{S}_{n}^{\theta}$.}
\end{figure}

\begin{figure}[htbp]
\centering
\includegraphics[width=\figwidth]{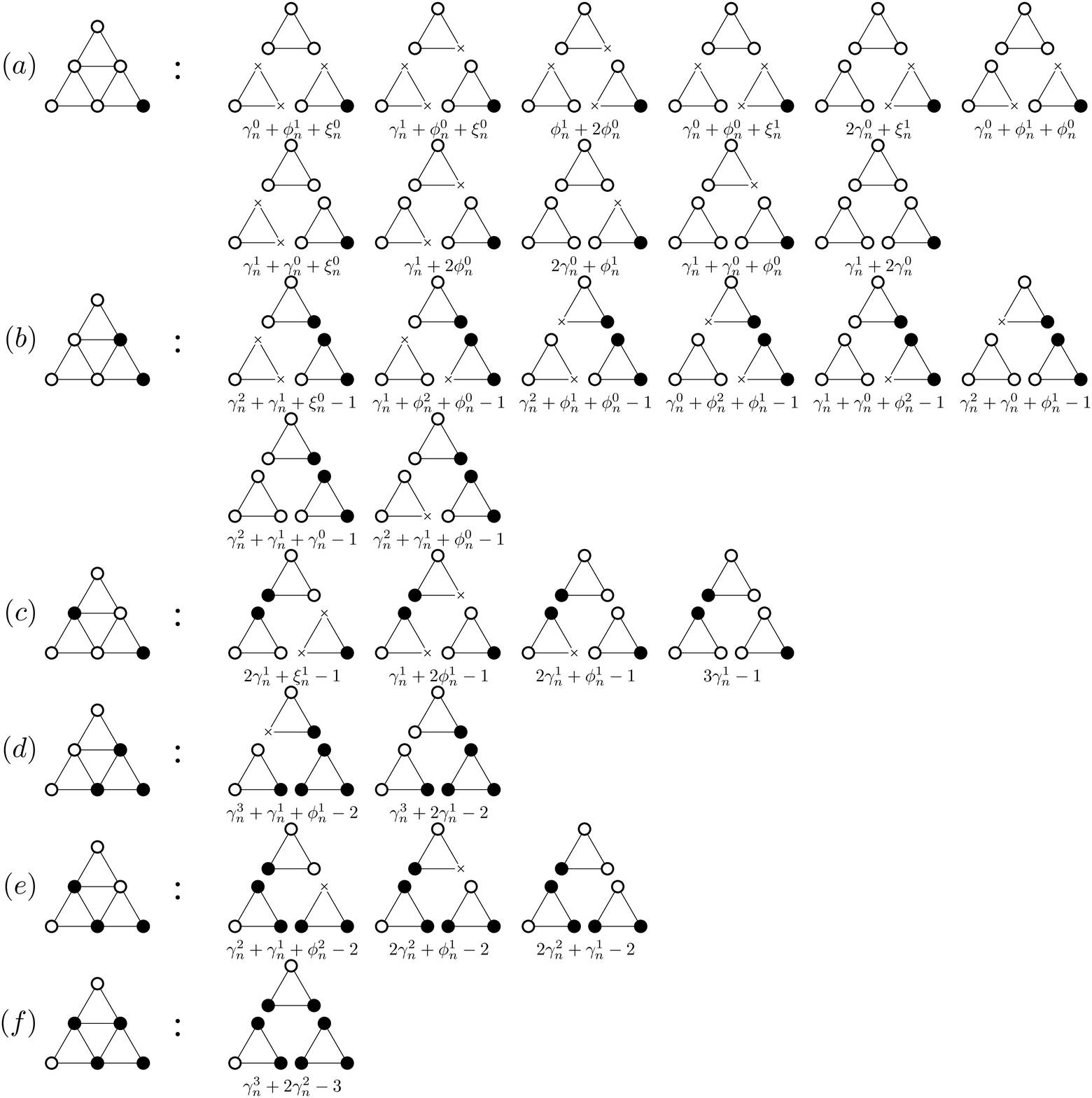}
\caption{\label{SGTheta1}Illustration of all possible configurations of dominating sets $\Omega_{n+1}^1$ of $\mathcal{S}_{n+1}$, which contain $\Theta_{n+1}^1$.}
\end{figure}

\begin{figure}[htbp]
\centering
\includegraphics[width=\figwidth]{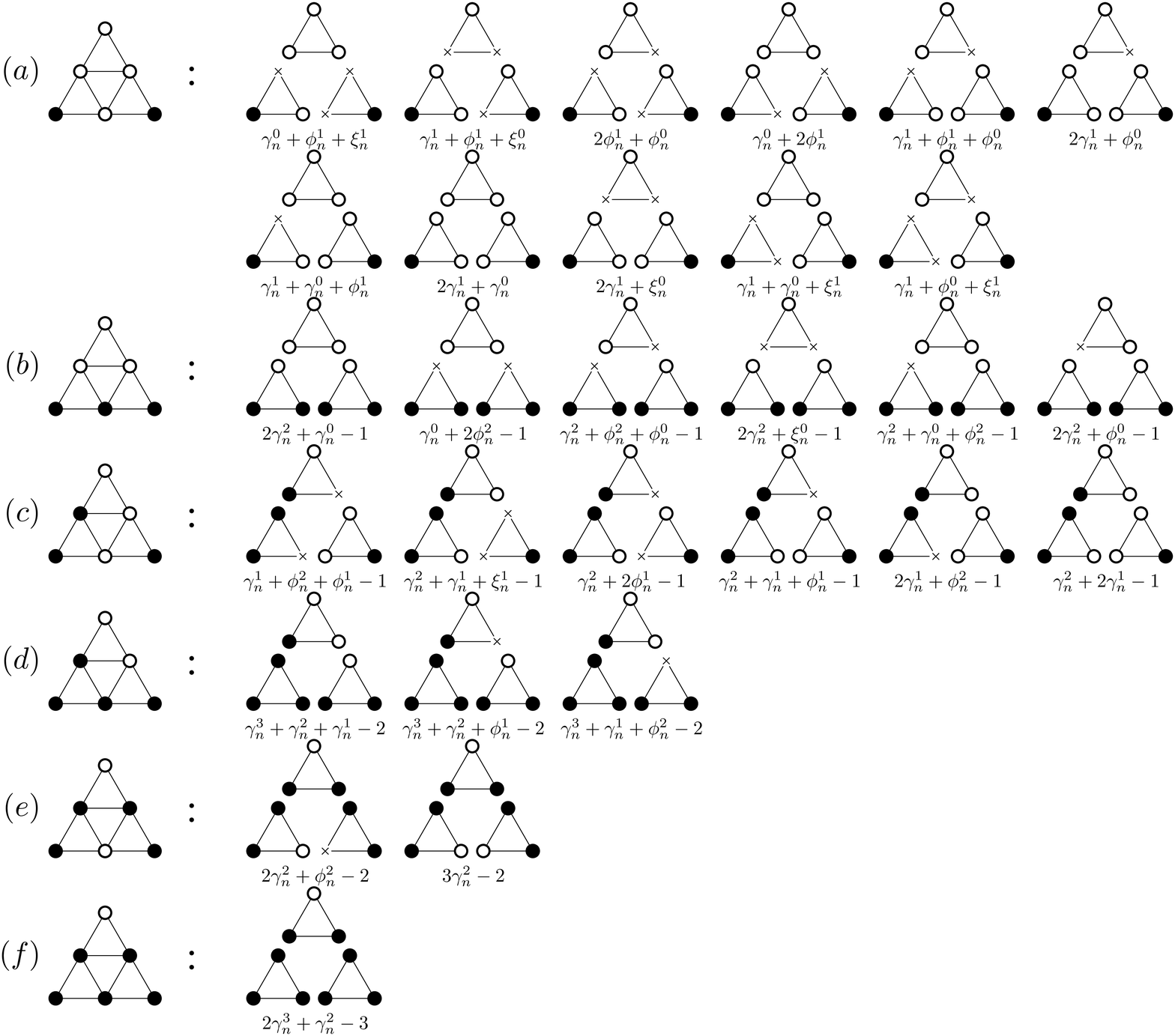}
\caption{\label{SGTheta2} Illustration of all possible configurations of dominating sets $\Omega_{n+1}^2$ of $\mathcal{S}_{n+1}$, which contain $\Theta_{n+1}^2$.}
\end{figure}

\begin{figure}[htbp]
\centering
\includegraphics[width=\figwidth]{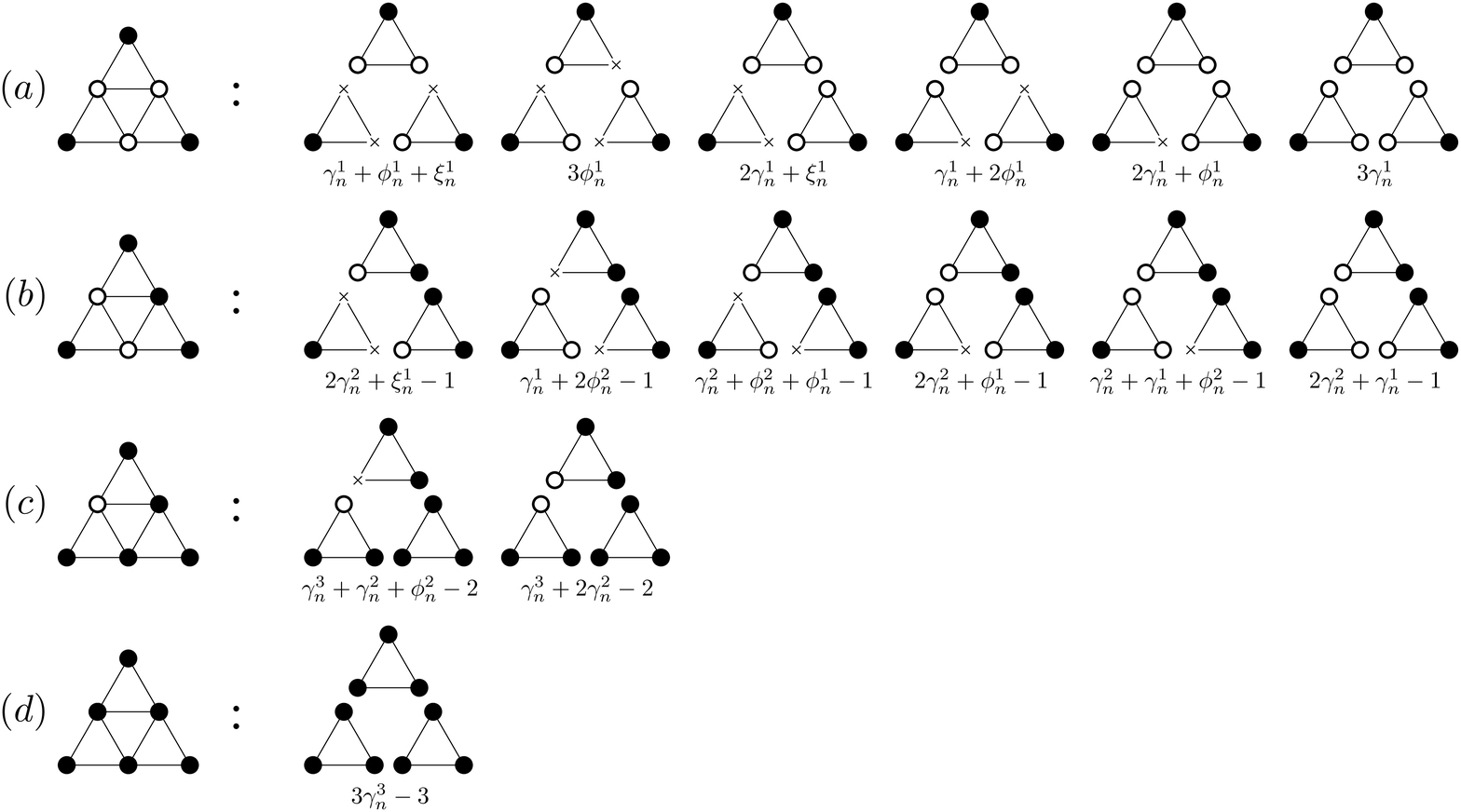}
\caption{\label{SGTheta3}Illustration of all possible configurations of dominating sets $\Omega_{n+1}^3$ of $\mathcal{S}_{n+1}$, which contain $\Theta_{n+1}^3$.}
\end{figure}

\begin{figure}[htbp]
\centering
\includegraphics[width=\figwidth]{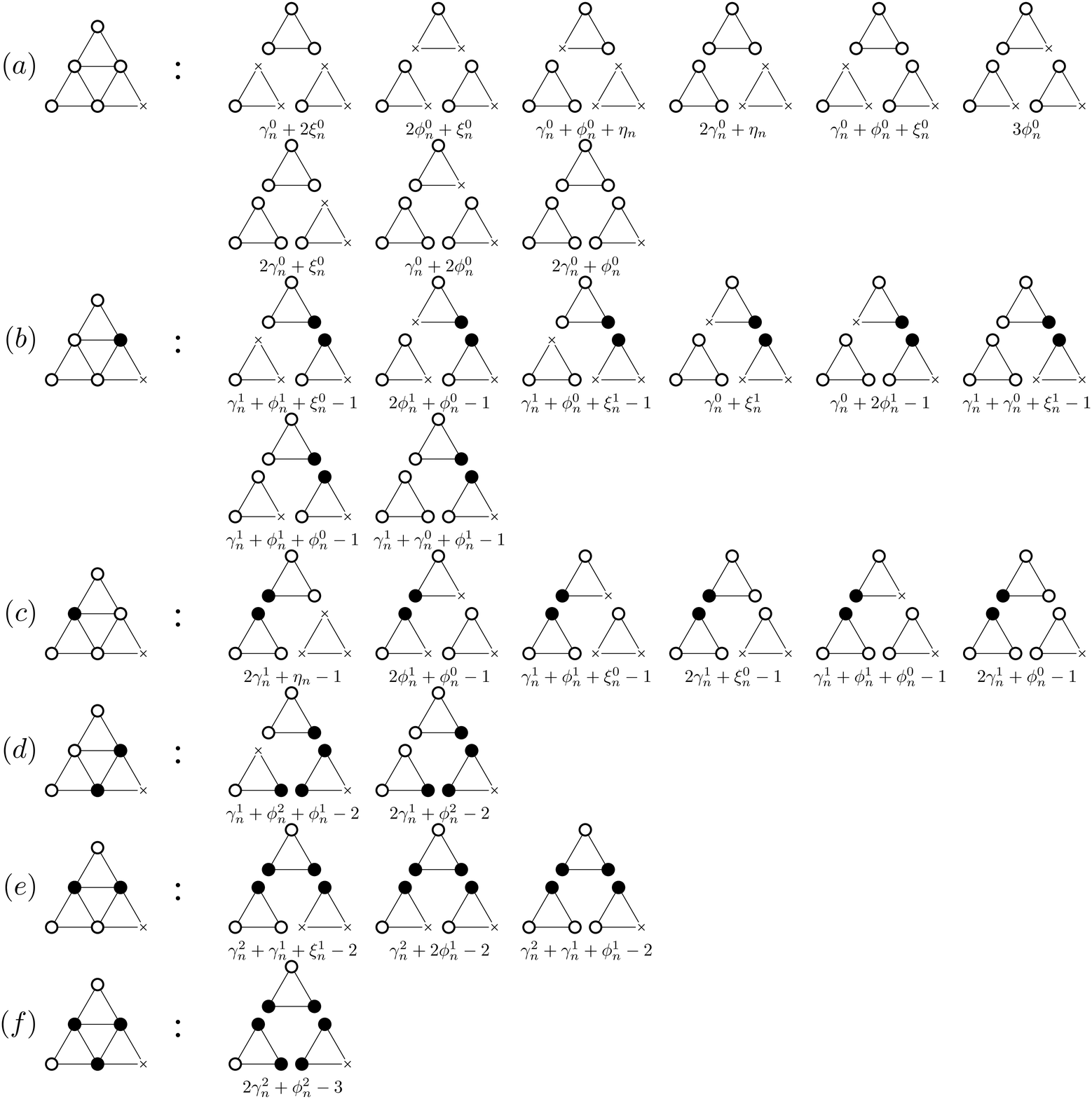}
\caption{\label{SGphi0}Illustration of all possible configurations of dominating sets of $\mathcal{S}_{n+1}^1$, which contain those sets with minimum cardinality  $\phi_{n+1}^0$.}
\end{figure}

\begin{figure}[htbp]
\centering
\includegraphics[width=\figwidth]{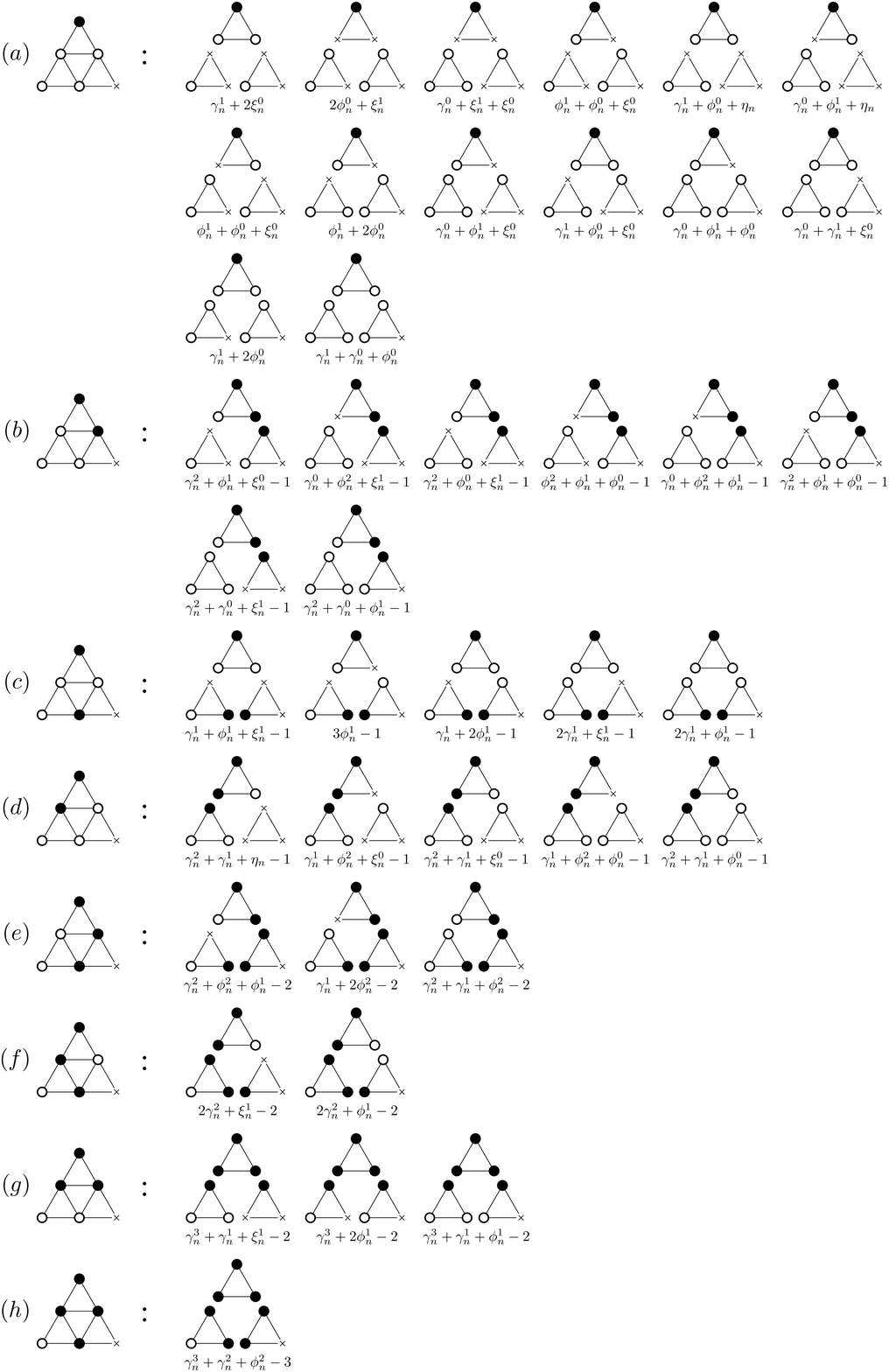}
\caption{\label{SGphi1}Illustration of all possible configurations of dominating sets of $\mathcal{S}_{n+1}^1$, which contain those sets with minimum cardinality  $\phi_{n+1}^1$.}
\end{figure}

\begin{figure}[htbp]
\centering
\includegraphics[width=\figwidth]{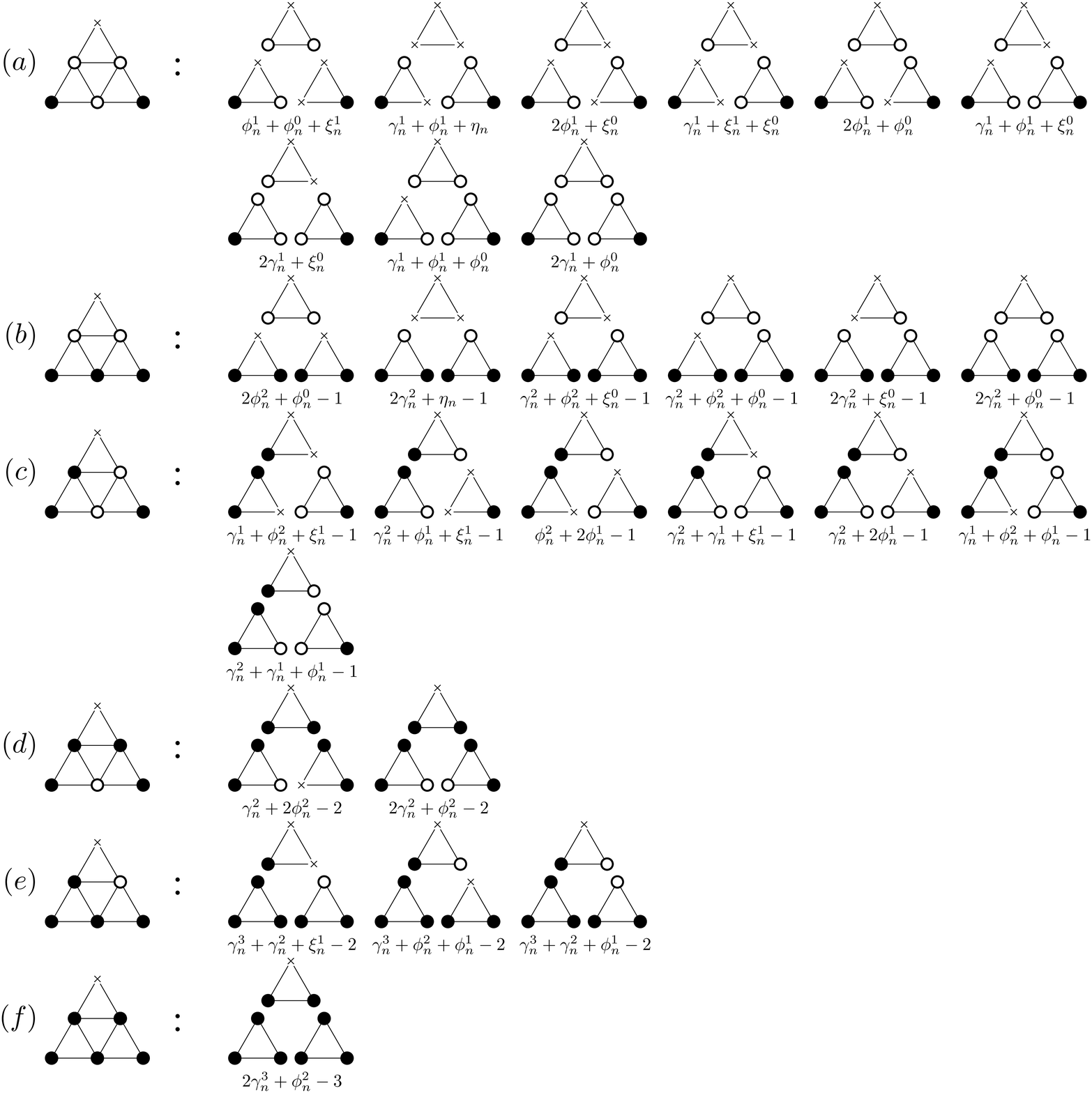}
\caption{\label{SGphi1}Illustration of all possible configurations of dominating sets of $\mathcal{S}_{n+1}^1$, which contain those sets with minimum cardinality $\phi_{n+1}^2$.}
\end{figure}

\begin{figure}[htbp]
\centering
\includegraphics[width=\figwidth]{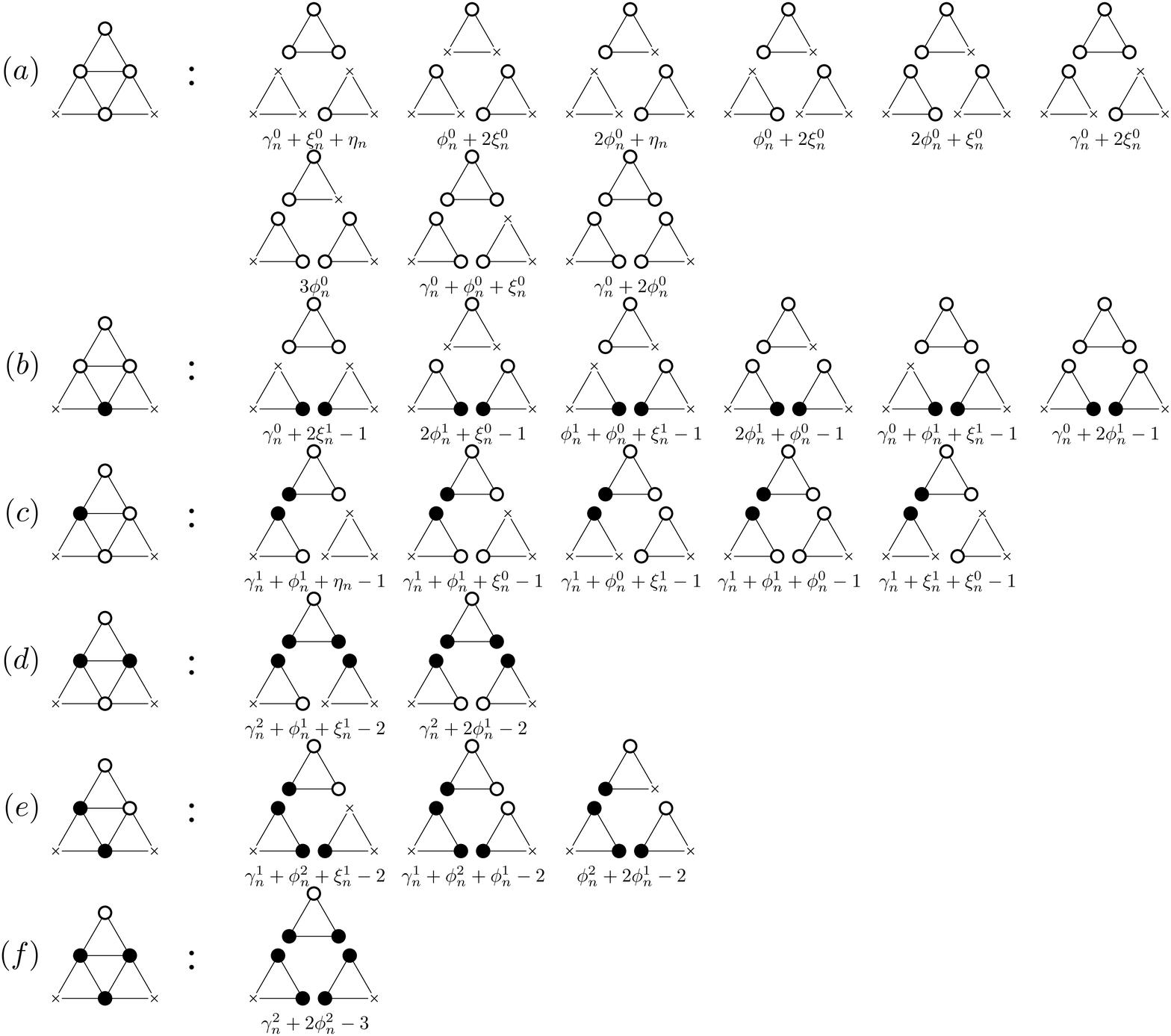}
\caption{\label{SGxi0}Illustration of all possible configurations of dominating sets of $\mathcal{S}_{n+1}^2$, which contain those sets with minimum cardinality $\xi_{n+1}^0$.}
\end{figure}

\begin{figure}[htbp]
\centering
\includegraphics[width=\figwidth]{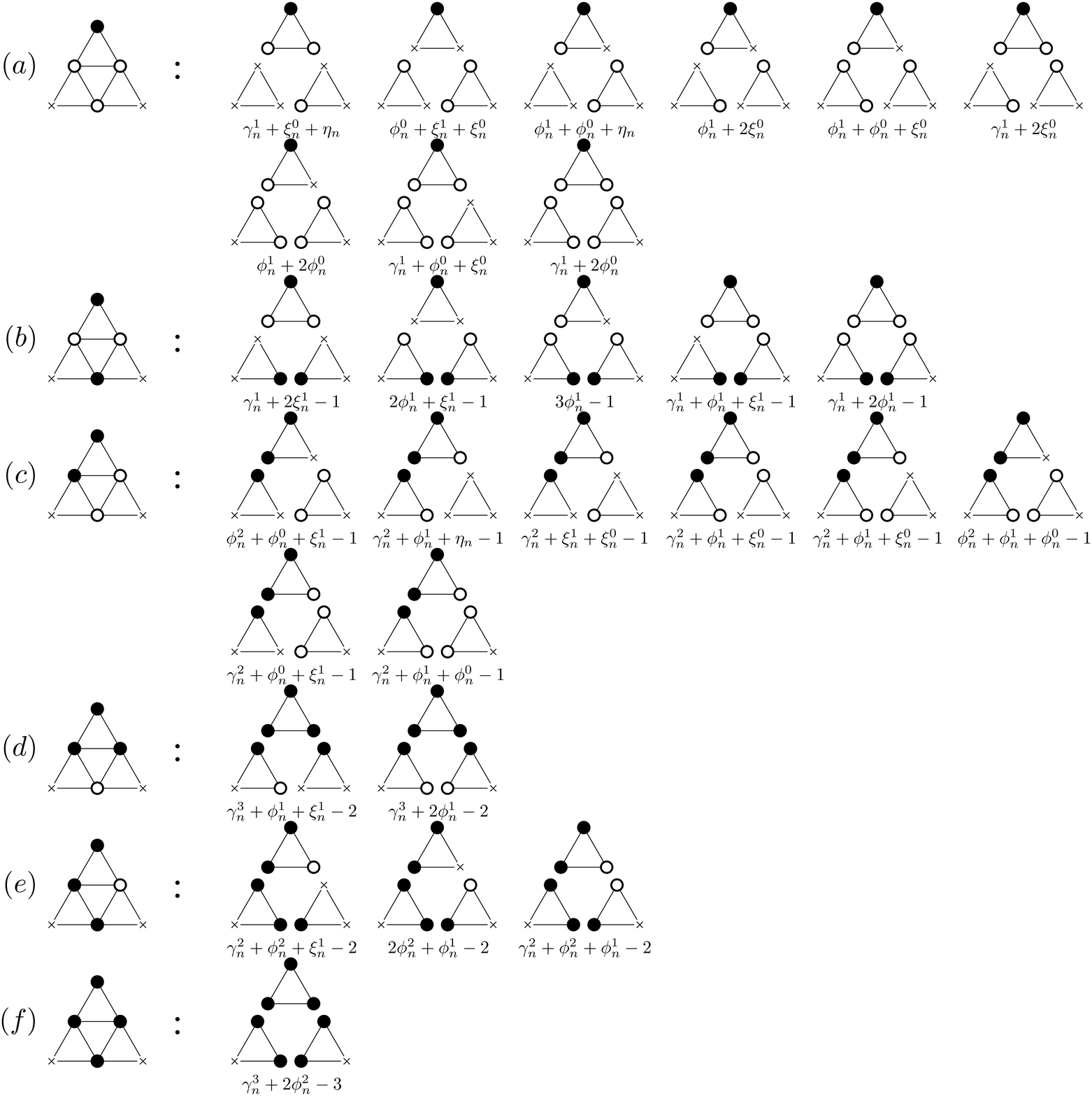}
\caption{\label{SGxi1}Illustration of all possible configurations of dominating sets of $\mathcal{S}_{n+1}^2$, which contain those sets with minimum cardinality $\xi_{n+1}^1$.}
\end{figure}

\begin{figure}[htbp]
\centering
\includegraphics[width=\figwidth]{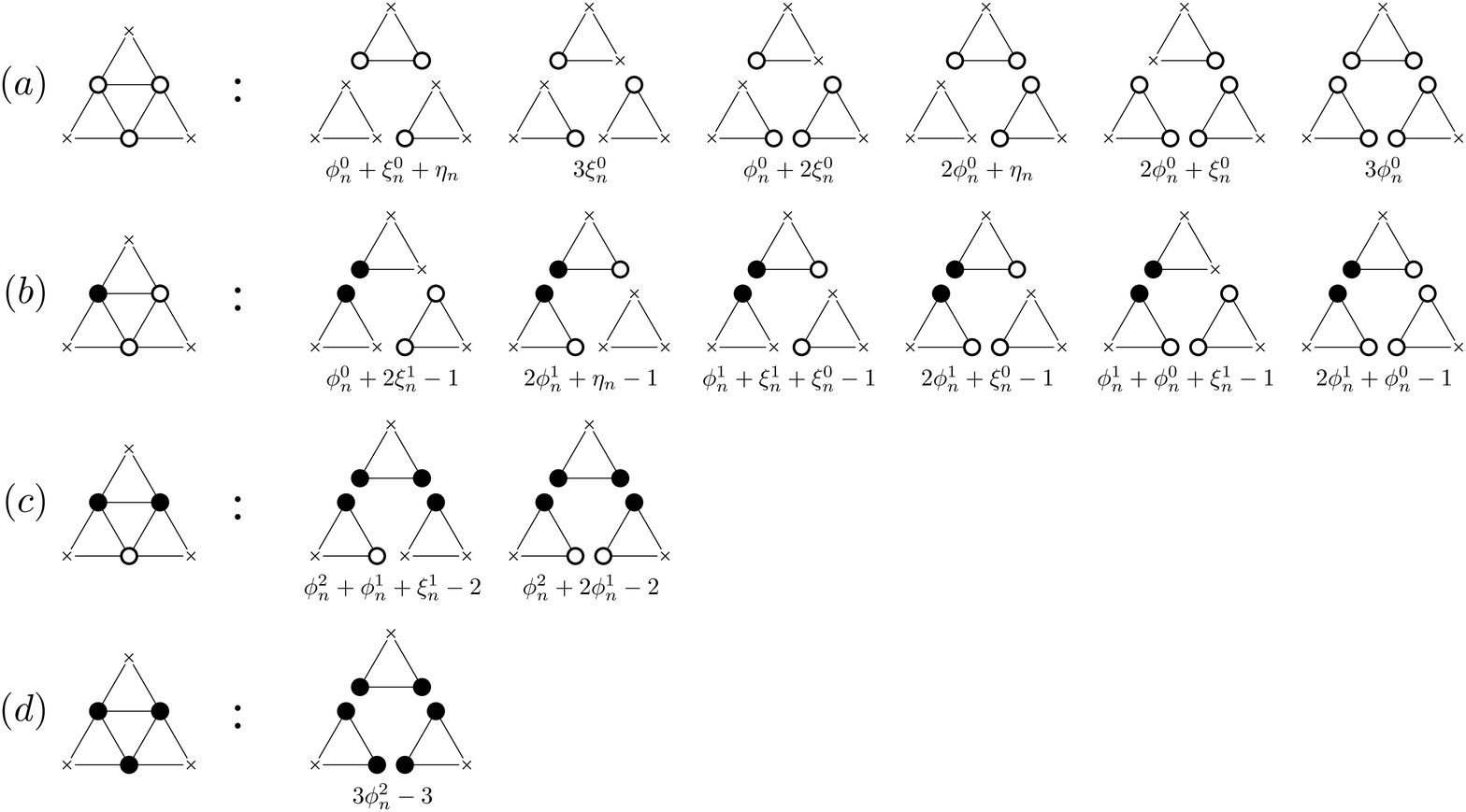}
\caption{\label{SGeta0}Illustration of all possible configurations of dominating sets of $\mathcal{S}_{n+1}^3$, which contain those sets with minimum cardinality $\eta_{n+1}$.}
\end{figure}



\begin{lemma}\label{leSGDom03}
For arbitrary $n \geq 3$,
\begin{align}\label{SGset01}
\gamma_n^1&= \gamma_n^0 +1\,, \nonumber \\
\gamma_n^2&=\gamma_n^0+2\,, \nonumber \\
\gamma_n^3&= \gamma_n^0+2\,, \nonumber \\
\phi_n^0 &= \gamma_n^0\,, \nonumber \\
\phi_n^1 &= \gamma_n^0+1\,, \nonumber \\
\phi_n^2 &= \gamma_n^0+1\,, \nonumber \\
\xi_n^0 &= \gamma_n^0\,, \nonumber \\
\xi_n^1 &= \gamma_n^0+1\,, \nonumber \\
\eta_n &= \gamma_n^0\,.  \nonumber
\end{align}
\end{lemma}
\begin{proof}
We prove this lemma by induction. \par For $n = 3$, it is easy to obtain by hand that $\gamma_3^0 =3$, $\gamma_3^1 = 4$, $\gamma_3^2=5$, $\gamma_3^3=5$, $\phi_3^0=3$, $\phi_3^1 = 4$, $\phi_3^2=4$, $\xi_3^0=3$, $\xi_3^1=4$, and $\eta_3=3$. Thus, the result is true for $n = 3$. \par
Let us assume that the lemma holds for $n = t$.  For $n =t+1$, by
induction assumption and lemma \ref{leSGDom02}, it is not difficult to check that the result is true for $n =t+1$. 
\end{proof}

\begin{theorem}
The domination number of the Sierpi\'nski graph $\mathcal{S}_n$, $n\geq3$, is $\gamma_n = 3^{n-2}$.
\end{theorem}
\begin{proof}
According  to Lemmas~\ref{leSGDom01}, \ref{leSGDom02}, and~\ref{leSGDom03}, we obtain  $\gamma_{n+1} = \gamma_{n+1}^0 = 3\gamma_n^0 = 3 \gamma_n$. Considering $\gamma_3 = 3$, it is obvious that $\gamma_n = 3^{n-2}$ holds for all $n \geq 3$.
\end{proof}

\subsection{Number of minimum dominating sets}

Let $a_n$ denote the number of minimum dominating sets of the Sierpi\'nski graph $\mathcal{S}_{n}$, $b_n$ the number of minimum dominating sets of  $\mathcal{S}_{n}^1$ containing no outmost vertices, $c_n$ the number of minimum dominating sets of  $\mathcal{S}_{n}^2$ including no outmost vertices,  $d_n$ the number of minimum dominating sets of  $\mathcal{S}_{n}^3$,  $e_n$ the number of minimum dominating sets of $\mathcal{S}_{n}^1$ containing two outmost vertices.

\begin{theorem}
For $n\geq 3$, the five quantities $a_n$, $b_n$, $c_n$, $d_n$ and $e_n$ can be obtained recursively according to the following relations.
\begin{equation}\label{SGset02}
a_{n+1} = 6a_nb_nc_n+2b_n^3+3a_n^2c_n+9a_nb_n^2+6a_n^2b_n+a_n^3,
\end{equation}
\begin{align}\label{SGset03}
b_{n+1} &= 2a_nc_n^2+4b_n^2c_n+2a_nb_nd_n+a_n^2d_n+8a_nb_nc_n+ \nonumber\\
&\quad 3b_n^3+2a_n^2c_n+4a_nb_n^2+a_n^2b_n,
\end{align}
\begin{align}\label{SGset04}
c_{n+1} &= 2a_nc_nd_n+4b_nc_n^2+2b_n^2d_n+7b_n^2c_n+3a_nc_n^2+\nonumber\\
&\quad 2b_n^3+2a_nb_nd_n+4a_nb_nc_n+a_nb_n^2,
\end{align}
\begin{equation}\label{SGset05}
d_{n+1} = 6b_nc_nd_n+c_n^3+9b_nc_n^2+3b_n^2d_n+6b_n^2c_n+b_n^3+e_n^3,
\end{equation}
\begin{equation}\label{SGset06}
e_{n+1} = b_ne_n^2,
\end{equation}
with the initial condition $a_3=2$, $b_3=3$, $d_3=2$, $d_3=1$ and $e_3=1$.
\end{theorem}
\begin{proof}
We first prove Eq.~\eqref{SGset02}. Note that $a_n$ is actually is the number of different minimum dominating sets for $\mathcal{S}_n$, each of which contains no outmost vertices.
Then, according to Lemma~\ref{leSGDom03} and Fig.~\ref{SGTheta0}, we can establish Eq.~\eqref{SGset02} by using the rotational symmetry of the Sierpi\'nski graph. \par

The remaining equations~\eqref{SGset03}-\eqref{SGset06} can be proved similarly.
\end{proof}

Applying Eqs.~\eqref{SGset02}-\eqref{SGset06}, the values of $a_n$, $b_n$, $c_n$, $d_n$ and $e_n$ can be recursively obtained for small $n$ as listed in Table~\ref{SetNo}, which shows that these quantities grow exponentially with $n$.

\begin{table*}
\caption{The first several values of $a_n$, $b_n$, $c_n$, $d_n$ and $e_n$.}\label{SetNo}
\normalsize
\centering
\begin{small}
\begin{tabular}{|c|c|c|c|c|c|c|}
\hline
\raisebox{-0.5ex}{n} & \raisebox{-0.5ex}{1} & \raisebox{-0.5ex}{2} & \raisebox{-0.5ex}{3} & \raisebox{-0.5ex}{4} & \raisebox{-0.5ex}{5} & \raisebox{-0.5ex}{6} \\[0.5ex]
\hline
\hline
\raisebox{-0.5ex}{$N_n$} & \raisebox{-0.5ex}{3} & \raisebox{-0.5ex}{6} & \raisebox{-0.5ex}{15} & \raisebox{-0.5ex}{42} & \raisebox{-0.5ex}{123} & \raisebox{-0.5ex}{366} \\[0.5ex]
\hline
\raisebox{-0.5ex}{$a_n$} & \raisebox{-0.5ex}{3} & \raisebox{-0.5ex}{6} & \raisebox{-0.5ex}{2} & \raisebox{-0.5ex}{392} & \raisebox{-0.5ex}{1,517,381,906} & \raisebox{-0.5ex}{84,494,691,003,170,101,058,068,575,600}\\[0.5ex]
\hline
\raisebox{-0.5ex}{$b_n$} & \raisebox{-0.5ex}{0} & \raisebox{-0.5ex}{1} & \raisebox{-0.5ex}{3} & \raisebox{-0.5ex}{381} & \raisebox{-0.5ex}{1,435,406,927} & \raisebox{-0.5ex}{79,618,813,236,624,661,173,376,634,785}\\[0.5ex]
\hline
\raisebox{-0.5ex}{$c_n$} & \raisebox{-0.5ex}{0} & \raisebox{-0.5ex}{0} & \raisebox{-0.5ex}{2} & \raisebox{-0.5ex}{356} & \raisebox{-0.5ex}{1,357,582,404} & \raisebox{-0.5ex}{75,023,197,813,382,628,339,656,804,330}\\[0.5ex]
\hline
\raisebox{-0.5ex}{$d_n$} & \raisebox{-0.5ex}{0} & \raisebox{-0.5ex}{0} & \raisebox{-0.5ex}{1} & \raisebox{-0.5ex}{315} & \raisebox{-0.5ex}{1,238,595,209} & \raisebox{-0.5ex}{68,189,726,461,267,338,496,884,215,735}\\[0.5ex]
\hline
\raisebox{-0.5ex}{$e_n$} & \raisebox{-0.5ex}{0} & \raisebox{-0.5ex}{1} & \raisebox{-0.5ex}{1} & \raisebox{-0.5ex}{3} & \raisebox{-0.5ex}{3429} & \raisebox{-0.5ex}{16,877,573,499,350,007}\\[0.5ex]
\hline
\end{tabular}
\end{small}
\end{table*}

\section{Comparison and analysis}

In the preceding two sections, we studied the minimum dominating set problem for the pseudofracal scale-free web and the Sierpi\'nski graph with identical number of vertices and edges. For both networks, we determined the domination number and the number of minimum dominating sets. From the obtained results, we can see that the domination number of the pseudofracal scale-free web is only half of that corresponding to the Sierpi\'nski graph. However, there exists only one minimum dominating set in pseudofracal scale-free web, while the number of different minimum dominating sets grows fast with the number of vertices. Because the size and number of minimum dominating sets are closely related to network structure, we argue that this distinction highlights the structural disparity between the pseudofracal scale-free web and the Sierpi\'nski graph.

The above-observed difference of minimum dominating sets between the pseudofracal scale-free web and the Sierpi\'nski graph can be easily understood. The pseudofracal scale-free web is heterogeneous, in which there are high-degree vertices that are connected to each other and are also linked to other small-degree vertices in the network. These high-degree vertices are very efficient in dominating the whole network. Thus, to construct a minimum dominating set, we will try to select high-degree vertices instead of small-degree vertices, which substantially lessens the number of minimum dominating sets. Quite different from the pseudofracal scale-free web, the Sierpi\'nski graph is homogeneous, all its vertices, excluding the three outmost ones, have the same degree of four and thus play a similar role in selecting a minimum dominating set. That's why the number of minimum dominating sets in the Sierpi\'nski graph is much more than that in the pseudofracal scale-free web. Therefore, the difference of minimum dominating sets between the Sierpi\'nski graph and the pseudofracal scale-free web lies in the inhomogeneous topology of the latter.

We note that although we only study the minimum dominating sets in a specific  scale-free  network, the domination number of which is considerably small taking up about one-sixth of the number of all vertices in the network. Since the scale-free property appears in a large variety of realistic networks, it is expected that the domination number of realistic scale-free networks is also much less, which is qualitatively similar to that of the pseudofracal scale-free web. Recent empirical study reported that in many real scale-free networks, and the domination number is smaller than the network size~\cite{NaAk12,NaAk13}.




\section{Conclusions}

To conclude, we studied the size and the number of minimum dominating sets in the pseudofracal scale-free web and the Sierpi\'nski graph, which have the same number of vertices and edges. For both networks, by using their self-similarity we determined the explicit expressions for the domination number, which for the former network is only one-half of that for the latter network. We also demonstrated that the minimum dominating set for the pseudofracal scale-free web is unique, while the number of minimum dominating sets in the Sierpi\'nski graph grows exponentially with  the number of vertices in the whole graph. The difference of the minimum dominating sets in the two considered networks is rooted in the their architecture disparity: the pseudofracal scale-free web is heterogeneous, while the Sierpi\'nski graph is homogeneous. Our work provides useful insight into applications of minimum dominating sets in scale-free networks.

\section*{Acknowledgements}

This work is supported by the National Natural Science Foundation of China under Grant No. 11275049. 

\end{document}